\documentclass[12pt,reqno]{amsart}
\usepackage{latexsym}
\usepackage{amsmath}
\usepackage{amssymb}
\usepackage{amsthm}
\usepackage{amsxtra}
\usepackage{amsfonts}
\usepackage{mathrsfs}
\usepackage{amsbsy}
\usepackage{graphicx}
\usepackage{t4phonet}
\renewcommand{\leq}{\leqslant}
\renewcommand{\geq}{\geqslant}

\newtheorem{teo}{Theorem}
\newtheorem{dfn}{Definition}[section]
\newtheorem{lem}{Lemma}[section]
\newtheorem{cor}{Corollary}[section]
\newtheorem{prop}{Proposition}[section]
\newtheorem{remark}{Remark}[section]

\newenvironment{dem}{\vspace{.2cm}\noindent {\bf Proof }\\}{\newline \hspace{1cm} \hspace{1cm}\flushright \hfill $\square$ \newline}

\newenvironment{dem2}[1]{\vspace{.2cm}\noindent {\bf Proof {#1}}\\}{\newline \hspace{1cm} \hspace{1cm}\flushright \hfill $\square$ \newline}

\renewcommand{\ker}{\operatorname{Ker}\,}

\newcommand{\EE}{\mathcal E}
\newcommand{\GG}{\mathcal G}

\newcommand{\DD}{\mathcal D}

\newcommand{\n}{\noindent}

\newcommand{\ve}{\varepsilon}
\newcommand{\erre}{\mathbb{R}} 
\newcommand{\ci}{\mathbb{C}}

\newcommand{\de}{\delta} 
\newcommand{\al}{\alpha}
\newcommand{\ga}{\gamma}
\newcommand{\la}{\lambda}

\newcommand{\ome}{\omega}
\newcommand{\f}{\frac}
\newcommand{\ba}{\begin{eqnarray}} \newcommand{\ea}{\end{eqnarray}}
\newcommand{\be}{\begin{equation}} \newcommand{\ee}{\end{equation}}
\newcommand{\bdm}{\begin{displaymath}} \newcommand{\edm}{\end{displaymath}} 
\newcommand{\brr}{\begin{array}}\newcommand{\err}{\end{array}}

\newcommand{\lf}{\left}
\newcommand{\ri}{\right}

\providecommand{\ove}[1]{\overline{#1}}

\newcommand{\bml}{\begin{gather}} 
\newcommand{\eml}{\end{gather}}
 
\DeclareMathOperator{\supp}{supp}

\DeclareMathOperator{\sech}{sech}
\DeclareMathOperator{\arctanh}{arctanh}
\newcommand{\beq}{\begin{equation}}
\newcommand{\eeq}{\end{equation}}
\newcommand{\RE}{\mathbb{R}}
\def\CO{{\mathbb C}}

\numberwithin{equation}{section}

\begin{document}

\title[]{Variational properties and orbital stability of standing waves for NLS equation on a star graph}

\author{Riccardo Adami}
\address{Dipartimento di Scienze Matematiche, Politecnico di Torino,  C.so Duca degli Abruzzi 24, 10129 Torino, Italy}
\email{riccardo.adami@polito.it}
\author{Claudio Cacciapuoti}
\address{Hausdorff Center for Mathematics,
  Institut f\"ur Angewandte Mathematik \\
Endenicher Allee, 60, 53115 Bonn, Germany}
\curraddr{Dipartimento di Scienza e Alta Tecnologia, Universit\`a dell'Insubria, Via Valleggio 11, 22100 Como, Italy}
\email{claudio.cacciapuoti@uninsubria.it        }
\author{Domenico Finco}
\address{Facolt\`a di Ingegneria, Universit\`a Telematica
Internazionale Uninettuno\\ 
Corso Vittorio E\-ma\-nue\-le II, 39, 00186 Roma,
Italy}
\email{d.finco@uninettunouniversity.net}
\author{Diego Noja}
\address{Dipartimento di Matematica e Applicazioni, Universit\`a
 di Milano
Bicocca \\
via R. Cozzi, 53, 20125 Milano, Italy}
\email{diego.noja@unimib.it}
\thanks{The authors would like to thank Reika Fukuizumi for several discussions. R.A., D.F., and D.N.  acknowledge the support of the FIRB 2012 project  ``Dispersive dynamics: Fourier Analysis and Variational Methods'' code RBFR12MXPO. R.A. was partially supported by the PRIN2012 grant ``Aspetti variazionali e perturbativi nei
problemi differenziali lineari''. C.C. acknowledges the support of the FIR 2013 project  ``Condensed Matter in Mathematical Physics''  code RBFR13WAET}

\begin{abstract}
We study standing waves for a nonlinear Schr\"odinger
equation on a star graph {$\mathcal{G}$} i.e.  $N$ half-lines joined at a
vertex. At the vertex an interaction occurs
described by a boundary condition of delta type with strength
$\alpha\leqslant 0$. The nonlinearity is of focusing power
type. The dynamics is given by an equation of the form $ i
\frac{d}{dt}\Psi_t = H \Psi_t - | \Psi_t |^{2\mu} \Psi_t $, where $H$
is the Hamiltonian operator which generates the linear Schr\"odinger dynamics. We show the existence of
several families of standing waves for every
sign of the coupling at the vertex
for every $\omega >
\frac{\alpha^2}{N^2}$. 
Furthermore,
we determine the ground
states, as minimizers of the
action on the Nehari manifold, and order the various
families. 
Finally, we show that
the ground states are orbitally stable for every allowed $\omega$ if the
nonlinearity is subcritical or critical, and for $\omega<\omega^\ast$ otherwise. 
 \par\noindent
\end{abstract}
\maketitle
\begin{small}
\n
\end{small}

\section{Introduction}
\noindent
In the present paper a rigorous analysis of the stationary behavior
of nonlinear Schr\"odin\-ger equation (NLS) on a graph is given, beginning
from the simplest type of unbounded graph, the star graph. In a
previous paper \cite{ACFN1} the authors studied the behavior in time
of an asymptotically solitary solution of NLS resident on a single
edge of the graph in the far past, and impinging on the vertex with
various types of couplings, giving a quantitative analysis of
reflection and transmission of the solitary wave after the collision
at the junction. Here we concentrate on a different phenomenon, namely
the
existence of persistent nonlinear bound states on the graph
(localized, or pinned nonlinear modes), and on their orbital
stability, when an attractive interaction is present at the vertex.  Some of the results here discussed and proved were briefly announced in \cite{ACFN2} .

\noindent
Let us briefly give a collocation of the model in the physical
context. Generally speaking, one can consider the NLS as a paradigm for the behavior of
nonlinear dispersive equations, but it is also an ubiquitous model
appearing in several concrete physical situations. The main fields of
application which we have in mind are the propagation of
electromagnetic pulses in nonlinear media (typically laser beams in
Kerr media or signal propagation in optical fibers), and dynamics of
Bose-Einstein condensates (BEC). We are interested in the way
solutions of NLS are affected by the presence of inhomogeneities of
various type. The propagation on the line in the presence of defects
has been a subject of intense study in the last years and it gives
rise to quite interesting phenomena, such as defect induced modes
\cite{CM,FMO,ANS}, i.e. standing solutions strongly localized around
the defect. The presence of defect modes affects propagation by
allowing trapping of wave packets, as experimentally shown in the case
of local photonic potentials in \cite{Linzon}. On the other hand,
nonlinearity can induce escaping of solitons from confining
potentials, as demonstrated in \cite{Pe}. A last interesting
phenomenon is the strong alteration of tunneling through potential
barriers in the presence of nonlinear defocusing optical media
\cite{Wan}. In this paper we consider NLS propagation through
junctions in networks. For example, when the dynamics of a BEC takes
place in essentially one-dimensional substrates (``cigar shaped"
condensates) or a laser pulse propagates in optical fibers and thin
waveguides, the question arises of the effect of a ramified junction
on propagation and on the possible generation of stable bound
states. The analysis of the behavior of NLS on networks is not yet a
fully developed subject, but it is currently growing. Concerning
situations of direct physical interest we mention the analysis of
scattering at Y junctions (``beam splitters") and other network
configurations (``ring interferometers") for one dimensional Bose
liquids discussed in \cite{TOD}. Some more results are known for the
discrete chain NLS model (DNLS), see in connection with the present
paper the analysis in \cite{Miro}. Other recent developments are in
\cite{GSD,Sob}. In particular, in the paper \cite{GSD} scattering from
a complex network sustaining nonlinear Schr\"odinger dynamics is
studied in relation to characterization of quantum chaos.\par\noindent
With these phenomenological and analytical premises in mind we would
like to construct a mathematical model capable to represent, in a
schematic but rigorous way, the propagation and stationary behavior
of a nonlinear Schr\"odinger field at a junction of a network.  We
begin by giving the needed preliminaries to rigorously define our
model. We recall that the {\it linear} Schr\"odinger equation on
graphs has been for a long time a very developed subject due to its
applications in quantum chemistry, nanotechnologies and more generally
mesoscopic physics. Standard references are
\cite{[BCFK06], [BEH], [Kuc04], [Kuc05], [KS99], EKKST08, BK13}, where more extensive
treatments are given.  Here we recall only the definitions needed to
have a self-contained exposition.  \n We consider a graph $\GG$
constituted by $N$ infinite half-lines attached to a common vertex.  The
natural Hilbert space where to pose a Schr\"odinger dynamics is then
$L^2(\GG)=\bigoplus_{j=1}^N L^2(\RE^+)$. Elements in $L^2(\GG)$ will
be represented as function vectors with components in $L^2(\RE^+)$,
namely
\[
\Psi=\left(
\begin{matrix}
\psi_1\\
\vdots \\
\psi_N
\end{matrix}
\right).
\]

\n
We denote the elements of $L^2(\GG)$ by capital Greek letters,
while functions in $L^2(\RE^+)$ are
 denoted by lowercase Greek letters.
We say that $\Psi$ is symmetric if $\psi_k$ does not depends on $k$.
The norm of $L^2$-functions on $\GG$ is naturally defined by
$$
\| \Psi \|^2_{L^2 (\GG)} : = \sum_{j=1}^N \| \psi_j \|^2_{L^2
  (\erre^+)} .
$$
From now on for the $L^2$-norm on the graph we
drop the subscript and simply write $\| \cdot \|$. Accordingly, we
denote by $(\cdot,\cdot)$ the scalar product in $L^2(\GG)$. 

\n
Analogously, given $1 \leqslant r \leqslant \infty$,
we define the space $L^r (\GG)$ as the set of functions
on the graph whose components are elements of the space $L^r (\erre^+)$,
and the norm is correspondingly defined by
\begin{equation*}
\big\|\Psi\big\|_{r}^{r}
=\sum_{j=1}^N\|\psi_j\|_{L^r(\RE^+)}^{r},
\ 1 \leqslant r < \infty, \qquad 
\big\|\Psi\big\|_{\infty}= \max_{1 \leqslant j \leqslant N}\|\psi_j\|_{L^\infty(\RE^+)} .
\end{equation*} 

\n
Besides, we need to introduce the spaces
$$
H^1(\GG) \equiv  \bigoplus_{j=1}^N   H^1(\erre^+)  \qquad
H^2(\GG) \equiv \bigoplus_{j=1}^N   H^2(\erre^+) 
, $$
equipped with the norms
\be \label{sobbo}
\| \Psi \|_{H^1}^2 \ = \ \sum_{i=1}^N \| \psi_i \|_{H^1(\erre^+)}^2,
\qquad
\| \Psi \|_{H^2}^2 \ = \ \sum_{i=1}^N \| \psi_i \|_{H^2(\erre^+)}^2.
\ee
Whenever a functional norm refers to a function defined on the graph,
we omit the symbol $\GG$. 

\n
When an element of $L^2 (\GG)$ evolves in time, we use in notation the
subscript $t$: for instance, $\Psi_t$. Sometimes we shall write $\Psi(t)$ in order
to highlight the dependence on time, or whenever such a notation is more
understandable.

\n
The dynamics we want to set on the graph is generated by a
linear part and a nonlinear one. We begin by describing the linear
part.\par\noindent 
Fixed $\alpha \in {\mathbb R}$, we consider a Hamiltonian operator,
denoted by $H$ and called $\delta$ graph or $\delta$ vertex, defined on the domain ${\mathcal D} (H)$ 
\begin{equation} \label{domdelta}
{\mathcal D} (H): =
\{ \Psi \in H^2(\GG) \text{ s.t. } \,
\psi_1 (0) = \ldots = \psi_N (0), \,\ \sum_{i=1}^N \psi_i ' (0) = \alpha \psi_1 (0) \},
\end{equation}
where $\psi'_i$ denotes the derivative of the function $\psi_i$ with
respect to the space variable related to the $i$-th edge. The
action of the operator $H$ is given by
\begin{equation*}
H \Psi =
\lf(
\begin{array}{c}
-\psi_1'' \\ \vdots \\ -\psi_N''
\end{array}
\ri).
\end{equation*}

\n
The Hamiltonian $H$ is a selfadjoint operator on $L^2(\GG)$ (\cite{[KS99]})
and generalizes to the graph the ordinary Schr\"odinger
operator with $\delta$ potential of strength $\alpha$ on the line
\cite{aghh:05}. Similarly to that case, the interaction is
encoded in the boundary condition. The case $\al=0$ in \eqref{domdelta}
plays a distinguished role and it defines what is usually given the name
of free or Kirchhoff boundary condition; we will indicate the
corresponding operator as $H^0$. Notice that for a graph with two
edges, i.e. the line, continuity of wavefunction and its derivative
for an element of ${\mathcal D}(H^0)$ makes the interaction disappear;
this fact justifies the name of free Hamiltonian. A $\delta$ vertex
with $\alpha<0$ can be interpreted as the presence of a deep
attractive potential well or attractive defect. This interpretation
can be enforced by showing that, as in the case of the line, the
operator $H$ is a norm resolvent limit for $\epsilon$ vanishing of a
scaled Hamiltonian $H_\epsilon=H^0+\alpha V_\epsilon$, where
$V_\epsilon=\frac{1}{\epsilon}V(\frac{x}{\epsilon})$ and $V$ is a
positive normalized potential on the graph (see \cite{[BEH]} and
reference therein). The attractive character shows in the fact that
for every $\alpha<0$ a (single) bound state exists for the linear
dynamics, with energy $-\frac{\alpha^2}{N^2}$. On the contrary, on a
Kirchhoff vertex no bound states exist, the spectrum is purely
absolutely continuous, but a zero energy resonance appears. Finally we
recall that in the case of repulsive delta interaction $\alpha>0$,
which is however of minor interest here, there are not bound states
nor zero energy resonances. 

\n
The quadratic form $E^{lin}$
associated to $H$ is defined on the finite energy space
\begin{equation*}
\EE \equiv{\mathcal D} ( E^{lin} ) = \{ \Psi \in H^1 (\GG)
\text{ s.t. } \,
\psi_1 (0) = \ldots = \psi_N (0) \}
\end{equation*}
and is given by
\begin{equation*}
E^{lin} [\Psi ] = \f{1}{2} \sum_{i=1}^N \int_0^{+\infty}
|\psi_i ' (x) |^2 \,dx\, + \f{\alpha}{2} | \psi_1 (0) |^2 = \frac{1}{2}\|\Psi'\|^2 + \f{\alpha}{2} | \psi_1 (0) |^2\ .
\end{equation*}
\n
The corresponding bilinear form is denoted by $B(\cdot,\cdot)$ and explicitly given by
\[   
B (\Psi, \Phi) \ : = \ \frac{1}{2}\sum_{i=1}^N( {\psi_i}^\prime, {\phi_i}^\prime)_{L^2(\RE^+)} + \frac{\alpha}{2} \overline{\psi_1} (0){\phi_1} (0).
\]

\n
As a particular case, the quadratic form $E^{0,lin}$ associated to $H^0$ is defined on
the  same space, that is ${\mathcal D} ( E^{0, lin} ) = \EE$,
and reads
\begin{equation*}
 E^{0,lin} [\Psi ] = \f{1}{2} \sum_{i=1}^N \int_0^{+\infty} |\psi_i ' (x) |^2 \,dx  = \frac{1}{2}\|\Psi'\|^2\ .
\end{equation*}



\par\noindent
Now let us introduce the nonlinearity. To this end we define
$G=(G_1,\dots, G_N):\CO^n\rightarrow \CO^n$ where $G$ acts
``componentwise" as $G_i(\zeta)=g(|\zeta_i|)\zeta_i$ for a suitable
$g:\erre^+\rightarrow \erre$ and $\zeta=(\zeta_i)\in
\CO^n$.\par\noindent 
We are interested in the special but important case of a power
nonlinearity of focusing type,  so we choose
$g(z)=-|z|^{2\mu},\ \mu >0 \ .$  

\n
After this preparation it is well defined the NLS equation on the graph, 
\beq
\label{diffform}
 i \frac{d}{dt}\Psi_t \ = \ H \Psi_t - | \Psi_t |^{2\mu} \Psi_t\,
\eeq
where $\mu >0$. This abstract nonlinear Schr\"odinger equation amounts
to a system of scalar NLS equations on the halfline, coupled through
the boundary condition at the origin included in the domain
\eqref{domdelta}. 

In Section 2 we show that for $\mu>0$
well-posedness of the dynamics described by equation \eqref{diffform}
(in weak form) for initial data in the finite energy space ${\mathcal
  E}$ holds true. Moreover, if $0<\mu<2$ then the solution exists for all
times and blow-up does not occur.  Finally, as in the standard NLS on
the line, mass $M(\Psi)=\frac{1}{2}\|\Psi\|^2$ and energy $E[\Psi]$
are conserved, where 
\[
 E[\Psi]=\frac{1}{2}\|\Psi'\|^2 -
\frac{1}{2\mu + 2}\|\Psi
\|^{2\mu+2}_{2\mu+2}+\frac{\alpha}{2}|\psi_1(0)|^2
\]
 and
analogously, in the case $\alpha = 0$, for the Kirchhoff energy
\[
E^0[\Psi]=\frac{1}{2}\|\Psi'\|^2 - \frac{1}{2\mu + 2}\|\Psi
\|^{2\mu+2}_{2\mu+2}\ .
\]
 \par After setting the model
and its well-posedness (see Section \ref{s:2}), we turn to the main subject of this paper,
existence and properties of standing wave solutions to
\eqref{diffform}. Standing waves are solutions of the form 
\begin{equation*}
\Psi_t(x)=e^{i\omega t}\ \Psi_{\omega}(x)\ .
\end{equation*}
The function
$\Psi_{\omega}$ is the amplitude or the profile (with some abuse of
interpretation) of the standing wave, and we will frequently refer to
the set of $\Psi_\omega$ as to the stationary states of the
problem.\par\noindent
The amplitude $\Psi_{\omega}$ satisfies the {\it stationary equation }
\[
H\Psi_{\omega} - |\Psi_{\omega}|^{2\mu} \Psi_{\omega} = -\ome \Psi_{\omega}\ ,\qquad \ome>0\ .
\]
This equation has a variational structure.\par\noindent
Let us define the {\it action} functional 
$$
S_\ome[\Psi] = E[\Psi] + \ome  M[\Psi] = \f 1 2 \|\Psi'\|^2 + \f \ome
2 \| \Psi \|^2 - \f{1}{2\mu +2}  \|\Psi\|_{2\mu+2}^{2\mu+2}\, +\,
\frac{\alpha}{2}|\psi_1(0)|^2. 
$$
The Euler-Lagrange equation of the action is the stationary equation above. 
The action $S_{\omega}$, defined on the form domain ${\mathcal E}$ of
the operator $H$, is unbounded from below. Nevertheless, it is bounded
on the so called natural (or Nehari) constraint $ \{ \Psi \in \EE$
\text{ s.t. }$I_\ome [\Psi] =0\}$, where $I_{\ome} [\Psi] =  \| \Psi'
\|^2 - \| \Psi \|_{2 \mu+ 2}^{2\mu +2} +\ome\| \Psi \|^{2}  
+\alpha|\psi_1(0)|^2$. Note that
$I_{\omega}(\Psi_{\omega})=S'_{\omega}(\Psi_{\omega})\Psi_{\omega}\ ,$
and thus 
the Nehari manifold is
a codimension one constraint which contains all the solutions to the
stationary equation. One of our main results is the following theorem.
\begin{teo}[Existence of minimizers for the Action functional] \mbox{ } \label{mainvar} 
\n
Let $\mu>0$. There exists $\alpha^\ast<0$ such that for $-N\sqrt{\ome} < \al< \al^\ast$
the action functional $S_\ome$  constrained to the Nehari manifold admits an absolute minimum, i.e. a $\Psi\neq 0$ such that $I_\ome [\Psi]=0$
and $S_\ome [\Psi] = {\rm inf} \left\{S_\omega[\Phi]\ :\ I_\omega[\Phi]=0 \right\}$.
\end{teo}
\noindent
So the action admits a constrained minimum on the natural constraint for
every $\omega>\frac{\alpha^2}{N^2}$ if the strength $\alpha$ of the
$\delta$ interaction at the vertex is negative and sufficiently
strong. The strategy of the proof, which is a consequence of results of Section 3 and Section 4, makes use of non trivial elements
and we give here some remarks. To get the existence of the minimum one
has at a certain point to compare the action $S_\omega$ with
$\alpha<0$ with the Kirchhoff action $S_{\omega}^0$. 
\par
In Section 3 we prove that the Kirchhoff action, while bounded from below on its natural
constraint, has no minimum (see \cite{ACFN3} for an analogous
phenomenon affecting the constrained energy functional). As a matter of fact, the infimum can be
exactly computed and it is achieved as the limit over a sequence of functions which escape at
infinity on a single edge. A main step in establishing the previous picture and in applying it to the $\delta$ case, is the exact calculation of the
infimum and the identification of the minimizers of the free action; to
this end one exploits an extension and generalization of the classical properties of
symmetric rearrangements of $L^p$ and $H^1$ functions to the case of
graphs.\par  
In Section 4 we prove Theorem 1. The analysis follows in part
proofs of similar results for singular interactions on the line given in \cite{FMO,[AN12]}, with major modifications due
to the fact that in this case the comparison with the free case is not
standard. In particular the upper bound in $\alpha$ given in the statement of Theorem 1 is a
consequence of the fact that one needs the condition ${\rm inf}
S_\omega<{\rm inf} S_\omega^0$ to guarantee the existence of an
absolute minimum in the constrained action, penalizing situations analogous to 
escaping minimizers of the free action. A sufficiently strong
attractive interaction at the vertex allows to satisfy the previous
condition. 

\noindent
We conjecture that the action has a {\it local} constrained
minimum that is larger than the infimum when the condition on $\alpha$ fails,
but presently we do not have a proof of this fact. 

In 
Section 5 an explicit construction of all the stationary states of the
problem is obtained, by solving the stationary equation for every value
of $\alpha$. It turns out that for every $N$ and $\omega>\frac{\alpha^2}{N^2}$ there exist
families $\{\Psi_{\omega,j}\}$ of stationary states of different
action and energy, which can be ordered in $j$ to form a nonlinear
spectrum (the family is unique only in the case $N=2$, i.e. the line). The
state of minimal action is the ground state $\Psi_{\omega,0}$, which is
of course the solution to the constrained minimum problem for the action just
discussed. The others are excited states, and they exists, for every $j$, when $\omega>\frac{\alpha^2}{(N-2j)^2}$. 
See Theorem 4 for a complete description.\par
Finally, in
Section 6 we study the stability of ground states. Stability is an
important requisite of a standing wave, because at a physical level
unstable states are rapidly
dominated by dispersion, drift or blow-up and so are undetectable (instability of NLS with a $\delta$ potential on the line is studied, partly numerically, in \cite{LeCFFKS}). The concept of stability,
due to gauge or $U(1)$ invariance of the action, is orbital
stability. The solutions remain close to the orbit
$e^{i\theta}\Psi_{\omega,0}$ of the ground state for all times if they
start close enough to it. The framework in which we study orbital
stability of the ground state is the mainstream of Weinstein and
Grillakis-Shatah-Strauss theory, which applies to infinite dimensional Hamiltonian systems such as abstract NLS equation when a regular branch of standing waves $\omega\mapsto\Psi_{\omega,0}$ (not necessarily ground states) exists, which is our case. \par\noindent According to this theory, to guarantee
orbital stability one needs to verify a set of spectral conditions on
the linearization of the NLS around the ground state, and a slope (or
Vakhitov-Kolokolov in the physical literature) condition concerning the behavior in $\omega$ of
the $L^2$-norm of the ground state.  Some adaptation of standard methods is needed to
treat the singular character of the interaction at the vertex, but in
fact it turns out that in the range of $\alpha$ over which Theorem $1$ is valid, the spectral conditions and the slope
condition are encountered for every $\omega$ and every nonlinearity in $\mu\in (0,2]$.
\begin{teo}[Orbital stability of the ground state] \mbox{} \label{orbitalstab}
 Let $\mu \in (0,2]$,
$\alpha<\alpha^*<0$, $\omega>\frac{{\alpha}^2}{N^2}$. Then the ground state
$\Psi_{\omega,0}\ $ is orbitally stable in ${\mathcal E}\ .$
\end{teo}
The proof of this result is contained in Section 6.
\n
Notice that one has orbital stability of the ground state in a
range of nonlinearities which includes the critical case. An
analogous phenomenon occurs in the case of the line, previously
treated in \cite{LeCFFKS}. This marks a difference with the case of a
free ($\alpha=0$) NLS on the line, where one has orbital instability
in the critical case. Finally, the proof of the previous theorem (see Remark 6.1) shows
that for supercritical nonlinearities $\mu>2$ the ground state is
orbitally stable for not too large $\omega$: there exists a threshold
$\omega^*>\frac{\alpha^2}{N^2}$ such that one has orbital stability
for the ground state $\Psi_{\omega,0}$ with $\frac{\alpha^2}{N^2} <
\omega < \omega^*$ and orbital instability in the opposite
case. 

Appendix A contains a theory of symmetric rearrangements on star graphs.
More precisely, the classical inequalities stating conservation of
$L^p$ norms and domination of kinetic energy are proved. This last
property, i.e. the P\'olya-Szeg\H{o} inequality, is particularly
interesting because it changes with respect to the case of the line
through the presence of a factor which takes into account the number of
edges of the graph, and this fact is crucial in the previously described analysis of action minimization on a star graph. 
A previous analysis of rearrangements on bounded graphs is contained in \cite{[Fri05]}, and a comparison of the two treatments is given at the end of Appendix A.
We stress the fact that the theory of rearrangements is a general tool and
it is in principle applicable to more general or different problems. 

\par We end this introduction with a few open problems and future
directions of study. Concerning technical issues, a different strategy
from the one here pursued in the analysis of ground states and their
stability is minimization of energy at constant mass (see the
classical paper \cite{[CL]} and for models related to the present one
\cite{ANV}); it requires a non trivial extension to graphs of
concentration-compactness method and it is studied in \cite{[ACFN4]}. Nothing is known up to now about stability
properties of the branches of excited states $\Psi_{\omega, j}$, which
exist for every $N>2$ and sufficiently high $\omega$; this is a
subject of special interest because there are only few cases where
excited states of NLS equations are explicitly known. The authors plan
to study this issue in a subsequent paper. Finally it would be interesting, and perhaps a difficult task, the extension of the analysis here given to different classes of graphs, possibly with non trivial topology. Several results in this direction were recently obtained in \cite{AST14-1, AST14-2, CFN14}. Dispersion properties, relevant to give precise large time behavior of solutions have been studied for trees, including star graphs, in \cite{BI1,BI2}. This is a first step for the analysis of possible asymptotic stability of standing waves on networks.
All these issues will need the development of new technical tools, both concerning variational analysis and stability properties.

\section{Well-posedness of  the model}
\label{s:2}
For our purposes it is sufficient to prove that the solution of the Schr\"odinger 
equation is uniquely defined in time in the energy domain and that
energy and mass are conserved quantities. This section is devoted to
the proof of these conservation laws and of  the  well-posedness of
equation \eqref{diffform}.  
\noindent
In fact along the proofs we shall always work with the weak form of \eqref{diffform}, namely
\begin{equation}
\label{intform1}
\Psi_t \ = \ e^{-iH t} \Psi_0 + i \int_0^t e^{-i
  H (t-s)}|\Psi_s|^{2\mu} \Psi_s \, ds\,.
\end{equation}
We consider the problem of the well-posedness in the sense of, e.g.,
\cite{[Caz]}, i.e., we prove  existence and uniqueness of the solution to equation \eqref{intform1}
in the energy domain of the system. Such a domain turns out
to coincide with the form domain of the linear part of equation
\eqref{diffform}. 
\noindent
We follow the traditional line of proving first  local
well-posedness,
and then extending it to all times by means of
a priori estimates provided by the
conservation
laws. 
\noindent
Proceeding as in \cite{ACFN1} where the cubic NLS is treated, we show
the well-posedness of the dynamics for any $\mu>0$, i.e. local existence
and 
uniqueness for initial data in the energy space. Moreover, we will
prove that  if $0<\mu<2$, then the well-posedness is
global, i.e. the solution exists for all times and no collapse occurs.
For a more extended treatment of the analogous problem for a 
two-edge vertex (namely, the real line with a point interaction at the
origin) see \cite{[AN09]}.

\n
We endow the energy domain $\EE$ with the $H^1$-norm defined in \eqref{sobbo}. 
Moreover  we denote by $\EE^\star$ the dual of $\EE$, i.e. the
set of the continuous linear functionals on $\EE$. 
We denote the dual product of $\Gamma \in \EE^\star$ 
and $\Psi \in\EE$
by
$
\langle \Gamma, \Psi \rangle$. In such a bracket we sometimes exchange the
place of the factor in $\EE^\star$ 
with the place of the factor in $\EE$: indeed, the duality product follows the same
algebraic rules of the standard scalar product.

\n
As usual, one can extend the action of $H$ to the space
$\EE$, with values in $\EE^\star$, by
$$
\langle H \Psi_1 , \Psi_2 \rangle \ : = B[\Psi_1, \Psi_2 ],
$$
where $B[\cdot,\cdot]$ denotes the bilinear form associated to the selfadjoint operator $H$.

\n
Furthermore, for any $\Psi \in \EE$
the identity
\begin{equation}
\label{extder}
\f d {dt} e^{-i H t} \Psi \ = \ - i H
e^{-i H t} \Psi
\end{equation}
 holds in $\EE^\star$
too. To prove it, one can first test the functional $\f d {dt} e^{-i
  H t} \Psi$ on an element $\Xi$ in the operator domain $\EE$, obtaining
\begin{equation*} \nonumber
\left\langle \f d {dt} e^{-i
  H t} \Psi, \Xi \right\rangle \ = \ \lim_{h \to 0} \left( \Psi, \f {e^{i
  H (t+h)} \Xi - e^{i H t} \Xi} h \right)\ = \ ( \Psi, i H e^{i H
t} \Xi) \ = \
\langle -i H e^{-i H t} \Psi, \Xi \rangle.
\end{equation*}
Then, the result can be extended to $\Xi \in \EE$ by a density argument.

\n
Besides, by \eqref{extder}, the differential version \eqref{diffform} of the
Schr\"odinger equation holds in $\EE^\star$.

In order to prove a well-posedness result we
need to generalize standard one-dimensional Gagliardo-Nirenberg
estimates to graphs, i.e.
\begin{equation}
\label{gajardo}
\| \Psi \|_{p} \ \leqslant \ C \| \Psi^\prime \|^{\f 1 2 - \f 1 p}
\| \Psi \|^{\f 1 2 + \f 1 p},
\end{equation}
where the $C > 0$ is a positive constant which depends on the index $p$ only.
The proof of \eqref{gajardo}
follows immediately from the analogous estimates for functions of the
real line,
considering that any function in $H^1 (\erre^+)$ can be extended to an even function
in $H^1 (\erre)$, and applying this reasoning to each component of $\Psi$ (see also \cite[I.31]{MPF91}).

\begin{prop}[Local well-posedness in $\EE$] \mbox{ }
\label{loch2}

\n Let $\mu >0$. For any $\Psi_0 \in \EE$, there exists $T > 0$ such that the
equation \eqref{intform1} has a unique solution $\Psi \in C^0 ([0,T),
\EE )
\cap C^1 ([0,T), \EE^\star)$. Moreover, eq. \eqref{intform1} has a maximal solution $\Psi^{\rm{max}}$
defined on an interval of the form $[0, T^\star)$, and the following ``blow-up
alternative''
holds: either $T^\star = \infty$ or
$$
\lim_{t \to T^\star} \| \Psi_t^{\rm{max}} \|_{\EE}
\ = \ + \infty,
$$
where we denoted by $\Psi_t^{\rm{max}}$ the function $\Psi^{\rm{max}}$ evaluated at time $t$.
\end{prop}
\begin{proof}
We define the space
$
{\mathcal X} : = L^\infty ([0,T), \EE),$
endowed with the norm
$
\| \Psi \|_{\mathcal X} \ : = \ \sup_{t \in [0,T)} \| \Psi_t
\|_{\EE}.
$
Given $\Psi_0 \in \EE$, we define the map $G : {\mathcal X}
\longrightarrow {\mathcal X}$ as
$$
G \Phi : = e^{- i H \cdot} \Psi_0 + i \int_0^\cdot
e^{- i H (\cdot - s)} | \Phi_s |^{2\mu} \Phi_s \, ds.
$$
We first notice that the nonlinearity preserves the space $\EE$. Then by $|(|\phi|^{2\mu}\phi)' |\leq C |\phi|^{2\mu} |\phi'|$ and using H\"older and Gagliardo-Nirenberg inequalities, one
obtains
$$
\| | \Phi_s |^{2\mu} \Phi_s \|_{\EE} \ \leqslant \ C \| \Phi_s \|_{\EE}^{2\mu+1},
$$
so that
\begin{equation}
\label{contraz1}
\begin{split}
\| G \Phi \|_{\mathcal X} \ \leqslant \ & \| \Psi_0 \|_{\EE} + C \int_0^T
\| \Phi_s \|_{\EE}^{2\mu+1} \, ds \
\leqslant \ \| \Psi_0 \|_{\EE} + C T \| \Phi \|_{\mathcal X}^{2\mu+1} \,.
\end{split}
\end{equation}
Analogously, given $\Phi, \Xi \in \EE$, one has 
\begin{equation}
\label{contraz2}
\begin{split}
\| G \Phi - G \Xi \|_{\mathcal X} \ \leqslant \ & C T
\left( \| \Phi \|_{\mathcal X}^{2\mu} + \| \Xi \|_{\mathcal X}^{2\mu} \right)
\| \Phi - \Xi \|_{\mathcal X}\,.
\end{split}
\end{equation}
We point out that the constant $C$ appearing in \eqref{contraz1} and
\eqref{contraz2} is independent of $\Psi_0$, $\Phi$, and $\Xi$.
Now let us restrict the map $G$ to elements $\Phi$ such that $\| \Phi
\|_{\mathcal X} \leqslant 
2 \| \Psi_0 \|_{\EE}$. From \eqref{contraz1} and \eqref{contraz2}, if
$T$ is chosen to be strictly less than $(8C \| \Psi_0 \|_{\EE}^{2\mu})^{-1}$, then
 $G$ is a contraction of the ball in ${\mathcal X}$ of radius
$ 2 \| \Psi_0 \|_{\EE}$, and so,
by the contraction lemma,
there exists a unique solution to \eqref{intform1} in the
time
interval $[0, T)$. By a standard one-step bootstrap argument one
immediately has that the solution actually belongs to $C^0 ([0,T),
\EE)$, and
due to the validity of \eqref{diffform} in the space
$\EE^\star$ we immediately have that the solution
$\Psi$ actually belongs to $C^0 ([0,T), \EE)) \cap
C^1 ([0,T),\EE^\star)$.

The proof of the
existence of a maximal solution is standard, while
the blow-up alternative is a consequence of the fact that,
whenever the $\EE$-norm of the solution is
finite, it is possible to extend it for a further time by the same
contraction
argument.
\end{proof}

The next step consists in the proof of the conservation laws.
\begin{prop}[Conservation laws] \mbox{ }

\n Let $\mu>0$. For any solution $\Psi \in C^0 ([0,T), \EE)
\cap C^1 ([0,T), \EE^\star)$ to
the problem \eqref{intform1}, the following conservation laws hold at
any time $t$:
\begin{equation*}
M[\Psi_t ] \ = M[ \Psi_0 ], \qquad
E[ \Psi_t ] \ = \ E[ \Psi_0 ].
\end{equation*}
\end{prop}

\begin{proof}
The conservation of the $L^2$-norm can be immediately obtained by the
validity of equation \eqref{diffform} in the space $\EE$:
$$
\f d {dt} M[\Psi_t]  \ = \, {\rm{Re}}  \, \left\langle \Psi_t ,
\f d {dt} \Psi_t \right\rangle \ = \ 0
$$
by the selfadjointness of $H$.
In order to prove the conservation of the energy, first we notice that
$\langle \Psi_t, H \Psi_t \rangle$ is differentiable as a function of time. Indeed,
\begin{equation*}
\begin{split} &
\f 1 h \left[ \langle \Psi_{t+h}, H \Psi_{t+h} \rangle -
\langle \Psi_{t}, H \Psi_{t} \rangle \right] 
\ = \ \left\langle \f{\Psi_{t+h} - \Psi_t} h, H \Psi_{t+h}
\right\rangle + \left\langle H \Psi_{t} ,\f{\Psi_{t+h} - \Psi_t} h
\right\rangle
\end{split}
\end{equation*}
and then, passing to the limit $h\to0$,
\begin{equation}
\label{previous}
\f d {dt} \left\langle \Psi_t ,
H \Psi_t \right\rangle \ = \ 2 \, {\rm{Re}} \, \left\langle \f d {dt} \Psi_t ,
H \Psi_t \right\rangle \ = \ - 2 \, {\rm{Im}} \, \langle | \Psi_t |^{2\mu}
\Psi_t , H \Psi_t \rangle,
\end{equation}
where we used the selfadjointness
of $H$ and \eqref{diffform}.
Furthermore,
\begin{equation}
\label{naechst}
\f d {dt} (\Psi_t ,
| \Psi_t |^{2\mu} \Psi_t) \ = \ \f d {dt} (\Psi_t^{\mu+1} , \Psi_t^{\mu+1})
 \ = \ -2 (\mu+1) \, {\rm{Im}} \,
\langle | \Psi_t |^{2\mu}
\Psi_t , H \Psi_t \rangle.
\end{equation}
From \eqref{previous} and \eqref{naechst} one then obtains
$$
\f d {dt} E[\Psi_t] \ = \ \f 1 2 \f d {dt} \langle \Psi_t ,
H \Psi_t \rangle - \f 1 {2\mu+2} \f d {dt} (\Psi_t ,
| \Psi_t |^2 \Psi_t)_{L^2} \ = \ 0
$$
and the proposition is proved.
\end{proof}
\begin{cor}[Global well-posedness] \mbox{ }

\n Let $0<\mu<2$. For any $\Psi_0 \in \EE$,  the
equation \eqref{intform1} has a unique solution $\Psi \in C^0 ([0,\infty),
\EE )
\cap C^1 ([0,\infty), \EE^\star)$. 
\end{cor}

\begin{proof}
By estimate \eqref{gajardo} with $p = \infty$ and conservation of
the $L^2$-norm, there exists
a constant $C$, that depends on $\Psi_0$ only, such that
$$
E[\Psi_0] \ = E[\Psi_t] \ \geq \
\f 1 2 \| \Psi_t^\prime \|^2 - C \| \Psi_t^\prime \|^{\mu}
$$
Therefore a uniform (in $t$) bound on
$ \| \Psi_t^\prime \|$ is obtained. As a consequence,
one has that no blow-up in finite
time can occur, and therefore, by the blow-up alternative, the
solution
is global in time.
\end{proof}


\section{Variational Analysis: the Kirchhoff vertex}

\n
In this section we compute the infimum of the action functional for the Kirchhoff case. 
As often in this framework, the action functional is unbounded from below and we have to restrict it to the Nehari manifold, or
natural constraint manifold, in order to have a functional bounded from below.
The knowledge of the infimum of the constrained action will be
a key ingredient in the next section
in the proof of the main theorem.

\n
The strategy of the computation of the infimum is standard: first we derive a lower bound and then we show that 
this lower bound is optimal by means of a minimizing sequence. In the derivation of the lower bound symmetric 
rearrangements are used. Using this technique we can map the initial variational problem into a variational problem
with symmetric functions which can be reduced to a problem on the halfline providing the required estimate.

\n 
The minimizing sequence shows in fact that the constrained action exhibits a sort of spontaneous symmetry breaking
in the Kirchhoff case. That is, although the functional is symmetric, the minimizing sequence is localized 
on a single edge.

\n As defined in the introduction, in the Kirchhoff case the action
functional is given by
\[
S_\ome^0[\Psi] = E^0[\Psi] +\ome M[\Psi] = \f 1 2 \|\Psi'\|^2 + \f
\ome 2 \| \Psi \|^2 - \f{1}{2\mu +2}  \|\Psi\|_{2\mu+2}^{2\mu+2}, 
\]
while the Nehari functional $I^0_{\ome} $ reads
\[
I_{\ome}^0 [\Psi] =  \| \Psi' \|^2 - \| \Psi \|_{2 \mu+ 2}^{2\mu +2} +\ome\| \Psi \|^{2} .
\]
The Nehari manifold is defined by $ \{ \Psi \in \EE,\; \Psi\neq 0 \text{ s.t. }I_\ome^0 [\Psi] =0\}$.
The action restricted to the Nehari manifold will be named reduced
action and is given by
\be\label{stilde0}
\widetilde S [\Psi]= S_\ome^0 [\Psi]-\f{1}{2}   I_{\ome}^0 [\Psi] = \frac{\mu}{2\mu+2} \| \Psi \|_{2 \mu+ 2}^{2\mu +2}.
\ee
It is understood that the domain of all the functionals is always $\EE$.
\begin{teo}[Infimum of the Action for the Kirchhoff case] \mbox{}

\label{actionkir}
\n
The infimum of the action functional $S^0_{\ome}$ restricted to the Nehari manifold is given by:
\be
\inf \{ S_\ome^0 [\Psi] \text{ s.t. } \Psi\in \EE,\,\Psi\neq 0 , \, I_\ome^0 [\Psi] =0\} \equiv
d^0 (\ome) =  (\mu+1)^{\frac{1}{\mu} } \ome^{ \frac{1}{\mu} + \frac{1}{2} }
\int_0^1 (1-t^2)^{\frac{1}{\mu}} dt .\label{sinf}
\ee
\end{teo}
\begin{dem}
The proof of \eqref{sinf} is divided into two parts: first we derive a lower bound for $S^0_\ome$, then we prove that 
the lower bound is optimal by means of a minimizing sequence.

\n
In order to derive a lower bound, we consider an auxiliary variational problem with symmetric functions.
This is done by using the rearrangements on the graph which are discussed in Appendix \ref{rearra}.

\n
Let $\Phi\in \EE$ and let $\Phi^\ast $ be its symmetric rearrangement.
We known that $\Phi^\ast$ is positive, symmetric and $\Phi^\ast \in \EE$. Moreover, 
by Theorem \ref{polya} and Proposition \ref{lp}, we have
\[
\| \Phi \| = \| \Phi^\ast \| 
\qquad 
\| \Phi \|_{2\mu+2} = \| \Phi^\ast \|_{2\mu+2}
\qquad
\| \Phi' \| \geqslant \f{2}{N} \| \Phi^{\ast '} \| .
\]
Therefore for $\Phi \in \EE $ such that $I^0_\ome [ \Phi]=0$ we have
\[
\f{4}{N^2} \| \Phi^{\ast \prime} \|^2 - \| \Phi^\ast \|_{2 \mu+ 2}^{2\mu +2} +\ome\| \Phi^\ast \|^{2} \leqslant 
I_\ome^0 [\Phi] =0
\]
and
\[
\widetilde S [\Phi] = \widetilde S [\Phi^\ast].
\]
\n
Taking into account \eqref{stilde0}, and the above properties of $\Phi^\ast$
one can enlarge the domain in the following way in order to lower the infimum,
\begin{multline} \label{tara}
\inf \{ S_\ome^0 [\Phi] \text{ s.t. } \Phi\in \EE,\,\Phi\neq 0 , \, I_\ome^0 [\Phi] =0\}=
\inf \{ \widetilde S [\Phi] \text{ s.t. } \Phi\in \EE,\,\Phi\neq 0 , \, I_\ome^0 [\Phi] =0\}\\
\geqslant
\inf \{ \widetilde S [\Phi] \text{ s.t. } \Phi\in \EE,\,\Phi\neq 0 , \, \Phi \text{ symmetric}, 
\,\f{4}{N^2} \| \Phi' \|^2 - \| \Phi \|_{2 \mu+ 2}^{2\mu +2} +\ome\| \Phi \|^{2} \leqslant 0
\}.
\end{multline}
Under the scaling, $\Phi(\cdot) \leadsto \la^{1/2} \Phi(\la \cdot)$,
$\la>0$, the last variational problem  scales as 
\begin{multline*}
\inf \lf\{ \widetilde S [\Phi] \text{ s.t. } \Phi\in \EE,\,\Phi\neq 0 , \, \Phi \text{ symmetric },\, 
\f{4}{N^2} \| \Phi' \|^2 - \| \Phi \|_{2 \mu+ 2}^{2\mu +2} +\ome\| \Phi \|^{2} \leqslant 0\ri\}= \\
\inf \la^\mu
\lf\{ \widetilde S [\Phi] \text{ s.t. } \Phi\in \EE,\,\Phi\neq 0 , \, \Phi \text{ symmetric },\, 
 \f{4}{N^2} \la^2\| \Phi' \|^2 - \la^\mu \| \Phi \|_{2 \mu+ 2}^{2\mu +2} +\ome\| 
\Phi \|^{2} \leqslant 0\ri\}.
\end{multline*}
It is convenient to choose $\la$ as
\[
\f{4}{N^2} \la^2 = \la^\mu \qquad {\mbox {so that}} \qquad \la=\lf( \f{N}{2} \ri)^{ \f{2}{2-\mu}} 
\]
in order to reconstruct a Nehari manifold with a rescaled $\ome$ as constraint. 
Moreover due to the symmetry of $\Phi$ we have
\begin{multline} \label{uni}
\lf( \f{N}{2} \ri)^{ \f{2\mu}{2-\mu}} \inf
\lf\{ \widetilde S [\Phi] \text{ s.t. } \Phi\in \EE,\,\Phi\neq 0 , \, \Phi \text{ symmetric }, 
 \| \Phi' \|^2 -  \| \Phi \|_{2 \mu+ 2}^{2\mu +2} +\ome \lf( \f{2}{N} \ri)^{ \f{2\mu}{2-\mu}} \| \Phi \|^{2} \leqslant 0\ri\} \\
=N \lf( \f{N}{2} \ri)^{ \f{2\mu}{2-\mu}}  \inf
\lf\{ \frac{\mu}{2\mu +2}  \|\phi\|^{2\mu +2}_{L^{2\mu +2} (\erre^+)}  \text{ s.t. } \phi\in H^1(\erre^+)\, ,
\ri.\\ \lf.
\phi\neq 0, \, 
 \| \phi' \|^2_{L^2(\erre^+)} -  \| \phi \|_{L^{2\mu +2} (\erre^+)}^{2\mu +2} 
+ \ome \lf( \f{2}{N} \ri)^{ \f{2\mu}{2-\mu}} \| \phi \|^{2}_{L^2(\erre^+ )} \leqslant 0\ri\} .
\end{multline}
It is convenient to introduce a variational
problem on the half line and an auxiliary variational problem on the line. Let $d^{\text{half}}(\ome )$ and 
$d^{\text{line}} (\ome)$ be defined 
in the following way:
\begin{align*}
& d^{\text{half}} (\ome )= 
\inf \lf\{ 
\frac{\mu}{2\mu +2}  \|\phi\|^{2\mu +2}_{L^{2\mu +2} (\erre^+)}  
\text{ s.t. }
\phi\in H^1(\erre^+),\, \phi\neq 0, \,
 \| \phi' \|^2_{L^2(\erre^+)} -  \| \phi \|_{L^{2\mu +2} (\erre^+)}^{2\mu +2} +
 \ome \| \phi \|^{2}_{L^2(\erre^+ )} \leqslant 0\ri\}  \\
& d^{\text{line}} (\ome) = 
\inf \lf\{ \frac{\mu}{2\mu +2}  \|\phi\|^{2\mu +2}_{L^{2\mu +2} (\erre)}  \text{ s.t. }\phi\in H^1(\erre),
\, \phi\neq 0,  \,
 \| \phi' \|^2_{L^2(\erre)} -  \| \phi \|_{L^{2\mu +2} (\erre)}^{2\mu +2} +
 \ome \| \phi \|^{2}_{L^2(\erre )} \leqslant 0\ri\}  .
\end{align*}
Notice that the following inequality holds true:
\be \label{tako}
2 d^{\text{half}} (\ome ) \geqslant d^{\text{line}} (\ome).
\ee 
Indeed, by absurd, assume $ 2 d^{\text{half}} (\ome ) < d^{\text{line}} (\ome)$
and let $\phi_n$ be a minimizing sequence for the problem on the halfline. 
We can extend $\phi_n$ by parity and obtain a sequence
$\tilde \phi_n \in H^1 (\erre) $ such that 
\[
\| \tilde \phi_n ' \|^2_{L^2(\erre)} -  \| \tilde \phi_n \|_{L^{2\mu +2} (\erre)}^{2\mu +2} +
 \ome \| \tilde \phi_n \|^{2}_{L^2(\erre )} \leqslant 0 \qquad \qquad \|\tilde \phi_n \|^{2\mu +2}_{L^{2\mu +2} (\erre)} = 2  \|\phi_n\|^{2\mu +2}_{L^{2\mu +2} (\erre^+)}.
 \]
Passing to the limit one would obtain 
\[
d^{\text{line}} (\ome) \geqslant \liminf_n \frac{\mu}{2\mu +2}  \| \tilde \phi_n \|_{L^{2\mu +2} (\erre)}^{2\mu +2} = 2 d^{{\rm half}} (\ome)
\]
which contradicts our absurd hypothesis. Therefore $1/2\,  d^{\text{line}} (\ome)$ provides a lower bound for the variational problem
we are interested in.
On the other hand the exact expression of $d^{\text{line}} (\ome)$ can be easily obtained from known results (see \cite{[Caz]} Ch. VIII) , and it is given by:
\be \label{ebi}
d^{\text{line}} (\ome)= 
 (\mu+1)^{\frac{1}{\mu} } \ome^{ \frac{1}{\mu} + \frac{1}{2} }
\int_0^1 (1-t^2)^{\frac{1}{\mu}} dt.
\ee
Taking into account \eqref{tara}, \eqref{uni}, \eqref{tako} and \eqref{ebi} we can conclude
\begin{align}\label{primaparte}
d^0 (\ome) & \geqslant 
\f{N}{2} \lf( \f{N}{2} \ri)^{ \f{2\mu}{2-\mu}} (\mu+1)^{\frac{1}{\mu} } 
\lf[\ome \lf( \f{2}{N} \ri)^{ \f{2\mu}{2-\mu}} \ri]^{ \frac{\mu+2}{2\mu} }
\int_0^1 (1-t^2)^{\frac{1}{\mu}} dt \nonumber \\
& =(\mu+1)^{\frac{1}{\mu} } \ome^{ \frac{\mu+2}{2\mu}  }
\int_0^1 (1-t^2)^{\frac{1}{\mu}} dt .
\end{align}

\n
Estimate  \eqref{primaparte} closes the first part of the proof. Now it is sufficient to exhibit a sequence of trial functions $\Phi_n$ satisfying 
the constraint and such that $ \widetilde S [\Phi_n ] \to (\mu+1)^{\frac{1}{\mu} } \ome^{ \frac{\mu+2}{2\mu}  }
\int_0^1 (1-t^2)^{\frac{1}{\mu}} dt$. 
We consider a sequence of soliton-like functions escaping
to infinity, i.e. 
\be \label{trial}
(\Phi_n)_i  ( x ) = 
\begin{cases} 
\phi_n (x) = \phi_s (x-n) \chi (x) & i=1\\
0 & i\neq 1
\end{cases}
\ee
where $\phi_s$ is defined in Appendix \ref{appB} by \eqref{soliton} and $\chi$ is a $C^\infty (\erre^+)$ function such that $0\leqslant \chi \leqslant 1$, $\chi=0$ for 
$0\leqslant x \leqslant 1$ and $\chi(x)=1$ for $x\geqslant 2$. The sequence $\Phi_n$ belongs to $\EE$ but does not satisfy the
constraint $I_\ome^0 [ \Phi_n]=0$. It is straightforward to check that 
\be \label{lowbou}
\| \Phi_n \|_{2\mu+2} \geqslant c
\ee
where the r.h.s. of \eqref{lowbou} depends on $\ome$ and $\mu$. In the remaining part 
of the proof we shall not make explicit the dependence on 
$\ome$ and $\mu$ of the constant appearing in estimates. Let $\de_n$ be defined by
\[
\de_n =
\lf(
\frac{ \|\Phi_n '\|^2 + \ome \|\Phi_n \|^2}{\| \Phi_n \|_{2\mu+2}^{2\mu+2} } \ri)^{ \f{1}{2\mu} }.
\]
It is straightforward to check that $I^0 [ \de_n \Phi_n] =0$. Then in order to prove \eqref{sinf},  it is sufficient to prove that
\begin{equation}
 \label{trialaction}
\lim_{n\to \infty} \widetilde S [ \de_n \Phi_n] = (\mu+1)^{\frac{1}{\mu} } \ome^{ \frac{\mu+2}{2\mu}  }
\int_0^1 (1-t^2)^{\frac{1}{\mu}} dt \ .
\end{equation}
Now we prove that 
\be \label{de1}
\lim_{n\to \infty} \de_n =1.
\ee
We have
\[
\de_n =
\lf( 1+ 
\frac{ \|\Phi_n '\|^2 + \ome \|\Phi_n\|^2 - \| \Phi_n \|_{2\mu+2}^{2\mu+2} }{\| \Phi_n \|_{2\mu+2}^{2\mu+2} } \ri)^{ \f{1}{2\mu} }
\]
and by \eqref{lowbou},  it is sufficient to prove that
\[
\lim_{n\to \infty} \|\Phi_n '\|^2 + \ome \|\Phi_n\|^2 - \| \Phi_n \|_{2\mu+2}^{2\mu+2} =0.
\]
Taking into account \eqref{trial} and \eqref{soliteq} and integrating by parts one has
\begin{multline*}
\lf|  \|\Phi_n '\|^2 + \ome \|\Phi_n\|^2 - \| \Phi_n \|_{2\mu+2}^{2\mu+2} \ri| \leqslant\\
c \int_0^\infty |\phi_s (x-n) |^2  |\chi ''(x)| |\chi (x)| dx + 
c \int_0^\infty |\phi_s  (x-n)|\phi_s '(x-n)|  |\chi '(x)||\chi (x)|  dx \equiv {\mathcal R}_n.
\end{multline*}
The remainder ${\mathcal R}_n$ can be estimated using the exponential decay of $\phi_s$ and $\phi_s '$ in the following way
\[
|{\mathcal R}_n| \leqslant c\int_{n-1}^\infty |\phi_s (x) |^2   dx + 
c \int_{n-1}^\infty |\phi_s  (x)\phi_s '(x)|    dx
\leqslant 
c \int_n^\infty e^{-c x} dx \leqslant c e^{-cn }.
\]
This proves \eqref{de1} while \eqref{trialaction} is reduced to prove that
\begin{equation*}
\lim_{n\to \infty} \widetilde S [ \Phi_n] = (\mu+1)^{\frac{1}{\mu} } \ome^{ \frac{\mu+2}{2\mu}  }
\int_0^1 (1-t^2)^{\frac{1}{\mu}} dt.
\end{equation*}
The last equality follows by dominated convergence and \eqref{formula2}:
\begin{multline*}
\lim_{n\to \infty} \widetilde S [ \Phi_n] = 
\lim_{n\to \infty} \frac{\mu}{2\mu +2} \int_0^\infty |\phi_s (x-n) \chi(x) |^{2\mu +2} dx =
\lim_{n\to \infty} \frac{\mu}{2\mu +2} \int_{-n}^\infty |\phi_s (x) \chi(x+n) |^{2\mu +2} dx = \\
 \frac{\mu}{2\mu +2} \int_{-\infty}^\infty |\phi_s (x)  |^{2\mu +2} dx= 
(\mu+1)^{ \f{1}{\mu} } \ome^{ \f{2+\mu}{2\mu} } \int_0^1 (1-t^2)^{ \f{1}{\mu} } dt .
\end{multline*}
The proof is concluded. 
\end{dem}
The previous proof shows that the infimum $d^0 (\ome)$ is approximated by the action of a soliton escaping to infinity.
Moreover notice that the minimizing sequence weakly converges to the vanishing function.


\section{Variational Analysis: the $\de$ vertex}
\n
In this section we discuss the variational properties of the action functional in 
the general case with $\al < 0$. In fact we prove that there exists
$\alpha^\ast < 0$ such that for $-N\sqrt{\ome}< \al
<\al^\ast $ the action functional constrained to the Nehari manifold
admits an absolute minimum. The proof of this statement is broken into several lemmas.
Firstly, in Lemma \ref{lemequi} we prove an equivalent formulation of the variational problem 
we are studying. Then, in Lemma \ref{lemlowbound} and Proposition \ref{iwashi} we prove that the infimum
of the constrained action is strictly positive and smaller than the
infimum of the Kirchhoff action, therefore
for $\al$ negative enough the infimum is not reached by functions escaping at infinity 
like \eqref{trial},
otherwise the two infima would coincide.
This is a key ingredient in the proof of the  main Theorem \ref{mainvar}, where we prove that a minimizing
sequence admits subsequences with non trivial weak limit. Finally we prove that this limit
is the absolute minimum.

\n
We recall from the introduction the action functional, given by
\begin{equation*} 
S_\ome[\Psi] = E[\Psi] +\ome M[\Psi]= S_\ome^0 [\Psi] +\f{\al }{2} |\psi_1 (0)|^2
\end{equation*}
and the Nehari functional 
\[
I_{\ome} [\Psi] =  \| \Psi' \|^2 - \| \Psi \|_{2 \mu+ 2}^{2\mu +2} +\ome\| \Psi \|^{2} 
+\alpha|\psi_1(0)|^2.
\]
The Nehari manifold is given by $\{ \Psi \in \EE, \, \Psi\neq 0, \text{ s.t. } 
I_\ome [\Psi] =0\}$.
It is understood 
that the above functionals are defined on the form domain $\EE$.
The action restricted to the Nehari manifold will be named again reduced action and is defined by
\be\label{stilde}
\widetilde S [\Psi] = \frac{\mu}{2\mu+2} \| \Psi \|_{2 \mu+ 2}^{2\mu +2}= S_\ome [\Psi]-\f{1}{2}   I_{\ome} [\Psi].
\ee
We introduce also the function
\[
d(\ome) = \inf \{ S_\ome [\Psi] \text{ s.t. } \Psi\in \EE,\,\Psi\neq 0 , \, I_\ome [\Psi] =0\}.
\]
In the following let $\al <0$. 

\n
We want to prove that for $\al$ smaller than a threshold value $\al^\ast$ the action $S_\ome$ constrained 
to the Nehari manifold admits an absolute minimum.

\n
Firstly we give an equivalent formulation of this variational problem.
\begin{lem} \label{lemequi}
The following equality holds
\be \label{equivalenza}
\inf \{ S_\ome [\Psi] \text{ s.t. } \Psi\in \EE,\,\Psi\neq 0 , \, I_\ome [\Psi] =0\}=
\inf \{\widetilde S [\Psi] \text{ s.t. } \Psi\in \EE,\,\Psi\neq 0 , \, I_\ome [\Psi] \leqslant 0\}.
\ee
Moreover $\Phi \in \EE$ satisfies $\widetilde S [\Phi] = d(\ome) $ and $I_\ome [ \Phi] \leqslant 0 $ iff
$S_\ome [\Phi] = d(\ome) $ and $I_\ome [\Phi]=0$.
\end{lem}
\begin{dem}
The idea is the following: if a function $\Phi$ is not on the Nehari manifold and $I_\ome [ \Phi] < 0 $,
then by multiplication by a suitable scalar, it can pulled on the manifold lowering the reduced action at the same time.
First notice that by \eqref{stilde} we have immediately 
\[
\inf \{ S_\ome [\Psi] \text{ s.t. } \Psi\in \EE,\, \Psi\neq0 , \,I_\ome [\Psi] =0\}\geqslant
\inf \{\widetilde S [\Psi] \text{ s.t. } \Psi\in {\mathcal E},\,\Psi\neq0 , \, I_\ome [\Psi] \leqslant 0\},
\]
since $S_\ome$ and $\widetilde S$ coincide on the Nehari manifold. 

\n
Now take $\Phi \in \EE$ such that $I_\ome [\Phi] <0$ and define 
\be \label{beta}
\beta = \lf( \f{ \|\Phi'\|^2 + \al |\phi_1 (0)|^2 + \ome \| \Phi \|^2 }{
\| \Phi \|_{2\mu + 2}^{2\mu + 2 } } \ri)^{\f{1}{2\mu} }.
\ee
Since $I_\ome [\Phi] <0$ then $\beta < 1$. Moreover by direct computation one has 
\[
I_\ome [ \beta \Phi] =0.
\]
Then, using again \eqref{stilde}, one has 
\[
S_\ome [ \beta \Phi] = \widetilde S [\beta \Phi] = \beta^{2\mu+2} \widetilde S [ \Phi] <
\widetilde S [\Phi]   
\]
then
\[
\inf \{ S_\ome [\Psi] \text{ s.t. } \Psi\in {\mathcal E},\, I_\ome [\Psi] =0\}\leqslant
\inf \{\widetilde S [\Psi] \text{ s.t. } \Psi\in {\mathcal E},\, I_\ome [\Psi] \leqslant 0\}
\]
and identity \eqref{equivalenza} has been proved.

\n
Notice that if $\Phi$ minimizes $S_\ome$ on $I_\ome =0$ then it minimizes also
$\widetilde S$ on $I_\ome \leqslant 0$ by \eqref{equivalenza}. Suppose now that
$\widetilde S [\Phi] = d(\ome) $ and $I_\ome [ \Phi] \leqslant 0 $. Then defining
$\beta$ as above one has that $S_\ome [ \beta \Phi]< \widetilde S [ \Phi] = d(\ome)$
which is  a contradiction to the definition of $d(\ome)$.
\end{dem}
\begin{lem} \label{lemlowbound}
Assume $\ome > \al^2/N^2$. Then $d(\ome)$ is strictly positive.
\end{lem}
\begin{dem}
Firstly we derive an elementary Sobolev inequality for the halfline. 
Let $f\in H^1 (\erre)$ and denote by $\hat f (k)$ its Fourier
transform. Then,
we have
\[
|f(0)| \leqslant \f{1}{\sqrt{2\pi}} \int |\hat f (k)|\, dk =
 \f{1}{\sqrt{2\pi}} \int |\hat f (k)| \f{ (a k^2 + a^{-1} ) }{ (a k^2 + a^{-1} )}
\, dk .
\]
By Cauchy-Schwarz inequality we have
\be \label{stimina}
|f(0)|^2 \leqslant \f{1}{2} ( a \| f' \|^2_{L^2(\erre)} + a^{-1} \| f \|^2_{L^2(\erre)}).
\ee
If $\phi\in H^1 (\erre^+)$ we can extend it by parity to a function on the line and apply \eqref{stimina}.
In this way we finally have
\be \label{stiminamina}
|\phi (0)|^2 \leqslant a \| \phi' \|^2_{L^2(\erre^+)} + a^{-1} \| \phi \|^2_{L^2(\erre^+)}.
\ee
Now take $\Phi \in \EE$ then by \eqref{stiminamina}, we have
\[
|\phi_1 (0) |^2 = \frac{1}{N} \sum_{i=1}^N |\phi_i (0) |^2
\leqslant \frac{1}{N} \sum_{i=1}^N \lf(  a  \|\phi_i \|^2_{L^2 (\erre^+)} 
+ \f{1}{ a} \|\phi_i ' \|^2_{L^2 (\erre^+)}\ri) =
\f{ a}{ N} \|\Phi \|^2 + \f{1}{ aN} \|\Phi ' \|^2.
\]
Then, using again \eqref{equivalenza} and with a suitable choice of $a$ (it is possible due to the restriction on $\ome$) we have
\[
0\geqslant \lf( 1 -\f{|\al|}{a} \ri) \| \Phi '\|^2 + \lf( \ome -a |\al| \ri) \|\Phi\|^2 
-\|\Phi\|_{2\mu+2}^{2\mu+2}\geqslant c \|\Phi\|_{H^1}^2 -\|\Phi\|_{2\mu+2}^{2\mu+2}.
\]
By Sobolev type inequalities we arrive at 
\[
 c \|\Phi\|^{2}_{2\mu+2} - \|\Phi\|_{2\mu+2}^{2\mu+2} \leqslant 0
\]
which implies, for a non vanishing function $\Phi$,
\[
 \|\Phi\|_{2\mu+2}\geqslant c.
\]
Since on the Nehari manifold $S_\ome$ and $\widetilde S$ coincide, we must have $d(\ome) >0$.
\end{dem}

\n
For any $\ome >0$ 
define $\al^\ast$ such that $-N\sqrt{\ome} <\al^\ast <0$ and 
\be \label{alstar}
\int_0^1 (1-t^2)^{ \frac{1}{\mu} } dt = \frac{N}{2} \int_{\f{|\al^\ast|}{N\sqrt{\ome}}}^1 (1-t^2)^{ \frac{1}{\mu} } dt.
\ee
Notice that $\al^\ast$ is uniquely defined since the r.h.s. of \eqref{alstar} is a decreasing function 
of $|\al^\ast|$ whose range includes the value $\int_0^1 (1-t^2)^{ \frac{1}{\mu} } dt $.
\begin{prop} \label{iwashi}
Let $-N\sqrt{\ome} < \al <\al^\ast$. Then 
\be\label{katsuo}
d(\ome) < d^0 (\ome).
\ee
\end{prop}
\begin{dem}
In order to prove \eqref{katsuo}, it is sufficient to exhibit a trial function $\widetilde\Psi\in \EE$ such that $I_\ome [\widetilde\Psi]=0$ and 
$S_\ome [ \widetilde \Psi ] < d^0 (\ome)$. 
Let $a$ be defined as 
\[
a = \f{1}{\mu \sqrt{\ome} } \arctanh \lf( \frac{|\al|}{ N \sqrt{\ome}} \ri).
\]
Now we consider the symmetric trial function $\widetilde \Psi$ given by:
\[
(\widetilde \Psi)_i (x) = \phi_s (x+a)
\qquad i=1\ldots N
\]
where $\phi_s$ is defined by \eqref{soliton}.
By construction $\widetilde\Psi \in \EE$. Moreover it is straightforward to check that 
$\widetilde\Psi \in {\mathcal D} (H) $ and 
\be \label{bashi}
H \widetilde\Psi - |\widetilde\Psi |^{2\mu} \widetilde\Psi = -\ome \widetilde\Psi.
\ee
Multiplying both sides of \eqref{bashi} and integrating by parts, 
one checks that $I_\ome [\widetilde \Psi ] = 0$ that is $\widetilde\Psi$
satisfies the constraint. Therefore it is sufficient to evaluate the reduced action using \eqref{formula2}. 
One has
\[
S_\ome [ \widetilde\Psi ] = \f N 2(\mu+1)^{\frac{1}{\mu} } \ome^{ \frac{1}{\mu} + \frac{1}{2} }
\int_{\f{|\al|}{N\sqrt{\ome}}}^1 (1-t^2)^{\frac{1}{\mu}} dt
\]
and the condition $S_\ome [ \widetilde \Psi ]< d^0 (\ome)$ amounts to
\[
\frac{N}{2} \int_{\f{|\al|}{N\sqrt{\ome}}}^1 (1-t^2)^{\frac{1}{\mu}} dt < \int_{0}^1 (1-t^2)^{\frac{1}{\mu}} dt
\]
which holds true by the hypothesis $-N\sqrt{\ome} < \al <\al^\ast$ since
\[
\frac{N}{2} \int_{\f{|\al|}{N\sqrt{\ome}}}^1 (1-t^2)^{\frac{1}{\mu}} dt < 
\frac{N}{2} \int_{\f{|\al^\ast|}{N\sqrt{\ome}}}^1 (1-t^2)^{\frac{1}{\mu}} dt =
\int_{0}^1 (1-t^2)^{\frac{1}{\mu}} dt .
\]
\end{dem}

\n
Now we can finally prove Theorem \ref{mainvar}, as stated in the introduction.

\begin{dem2}{of Theorem \ref{mainvar}}
Let $\{\Psi_n \}$ be a minimizing sequence, we prove that there exists a subsequence weakly convergent in $H^1$.
First notice that $\|\Psi_n\|_{2\mu +2}$ is obviously bounded (see Lemma \ref{lemequi}). 
Recall that for $\Phi \in \EE$
\be \label{quadinf}
\| \Phi ' \|^2 + \al |\phi_1 (0)|^2 \geqslant \f {\al^2}{ N^2} \|\Phi\|^2 .
\ee
Using \eqref{quadinf} and $I_\ome [\Psi_n]\leqslant 0 $ we have
\[
0\leqslant \lf( \ome - \f{\al^2}{N^2} \ri) \| \Psi_n \|^2 \leqslant 
\| \Psi_n ' \|^2  +\ome\| \Psi_n \|^{2} +\alpha|\psi_{n,1}(0)|^2 \leqslant \| \Psi_n \|_{2 \mu+ 2}^{2\mu +2} 
\leqslant c .
\]
This implies $\| \Psi_n \|\leqslant c$.
Using again  $I_\ome [\Psi_n]\leqslant 0 $ we have also
\[
\| \Psi_n ' \|^2 \leqslant \| \Psi_n \|_{2 \mu+ 2}^{2\mu +2} - \ome\| \Psi_n \|^{2} - \alpha|\psi_{n,1}(0)|^2 \leqslant 
c + c \| \psi_{n,1}'\| \| \psi_{n,1}\|  \leqslant c + c\lf( \f{1}{\ve} \| \Psi_n\|^2 +\ve \| \Psi_n '\|^2 \ri)
\]
for any $\ve>0$. Taking $\ve$ sufficiently small we see that $\| \Psi_n ' \|$ is bounded and therefore 
also $\| \Psi_n ' \|_{H^1}$ is bounded. By Banach-Alaoglu theorem there exists a weakly convergent subsequence,
which will be still denoted by $\{ \Psi_n \}$. Let $\Psi_{\infty}$ be the weak limit.

\n
Now we prove that $\Psi_{\infty}\neq 0$. To this aim we preliminarily show that $\Psi_\infty \in {\mathcal E}$, 
$\Psi_n (0)\to \Psi_\infty (0)$ and 
that $I_\ome [ \Psi_n ] \to 0$. Let $\Lambda^j: \GG \to \erre$ be a function on the graph defined in the following way: 
$\la_j  (y) = e^{-y}$ and 
$\la_i (y)=0$ for $i\neq j$. Then by weak convergence and integration by parts we have
\be \label{halflimit}
\psi_{j,n} (0) = (\Lambda^j , \Psi_n )_{H^1} \to (\Lambda^j , \Psi_\infty )_{H^1}=\Psi_{j,\infty} (0).
\ee
Since $\psi_{j,n} (0) $ does not depend on $j$, the first two claims
are proved. We prove the last claim by contradiction. Assume that  
\be \label{claudio}
I_\ome [ \Psi_n ] \to 0
\ee
 is false, then there exists a subsequence, still denoted by $\{ \Psi_n \}$,
such that
\be \label{secondpre}
\lim_{n\to \infty} I_\ome [ \Psi_n ] = \ga <0 .
\ee
Let $\beta_n$ be defined according to \eqref{beta}
then
\[
\lim_{n\to \infty} \beta_n =
\lim_{n\to \infty} \lf( 1 + \f{ I_\ome [ \Psi_n ]}{ \| \Psi_n \|_{2 \mu+ 2}^{2\mu +2}} \ri)^{\f{1}{2\mu}}=
\lf( 1 + \f{\ga \mu}{2(\mu+1) d(\ome) } \ri)^{\f{1}{2\mu}} <1
\]
therefore
\[
\lim_{n\to \infty} \widetilde S [ \beta_n \Psi_n ] = \lim_{n\to \infty}\beta_n^{2\mu +2}  \widetilde S [ \Psi_n ] < d(\ome)
\]
and 
$I_\ome [ \beta_n \Psi_n ] = 0$ but this contradicts the assumption that $\Psi_n$ is a minimizing sequence.
Hence $I_\ome [ \Psi_n ] \to 0$.

\n
We proceed again by contradiction to prove that $\Psi_\infty \neq 0$. Assume that  $\Psi_\infty = 0$
and define 
\[
\rho_n = \f{\lf[ \| \Psi_n '\|^2 + \ome  \| \Psi_n\|^2\ri]^{\f{1}{2\mu} }}{\| \Psi_n \|_{2\mu +2}^{1+ \f{1}{\mu}} }.
\]
Using \eqref{halflimit}, \eqref{secondpre} and the contradiction hypothesis, one has
\[
\lim_{n\to \infty} \rho_n =
\lim_{n\to \infty} \lf( 1+ \f{ I_\ome [ \Psi_n ]  - \al |\psi_{1,n}
  (0) |^2 }{\| \Psi \|_{2\mu +2}^{2\mu +2}}\ri)^{\frac{1}{2\mu} } 
= 1.
\]
Therefore 
\[
\lim_{n\to \infty} \widetilde S[ \rho_n \Psi_n] = 
\lim_{n\to \infty}  \rho_n^{2\mu+2} \widetilde S[ \Psi_n] = d(\ome).
\]
On the other hand, by direct computation one has 
\[
I_\ome^0  [ \rho_n \Psi_n ] =0.
\]
Therefore, by Proposition \ref{iwashi} and Theorem \ref{actionkir}
\[
d(\ome) < S_\infty (\ome) \leqslant \widetilde S[ \rho_n \Psi_n].
\]
Passing to the limit, one obtains
\[
d(\ome) < S_\infty (\ome) \leqslant d(\ome)
\]
therefore the hypothesis $\Psi_\infty =0$ can not hold.

\n
Now we shall prove that $I_\ome [\Psi_\infty] \leqslant 0$. We recall, see \cite{BL}, Brezis and Lieb's lemma: if
$f_n$ converges weakly to $f_\infty$ in $L^p$, $1<p<\infty$, then
\be \label{fatou}
 \|f_n\|_p^p - \|f_n-f_\infty\|_p^p - \|f_\infty\|_p^p \to 0.
\ee
In our case, this implies that
\be \label{breli3}
\widetilde S [\Psi_n ]- \widetilde S [\Psi_n - \Psi_\infty ] - \widetilde S [\Psi_\infty ] \to 0
\ee
and, applying \eqref{fatou} both to $\Psi$ and $\Psi'$, that
\be \label{liebbata}
I_\ome [\Psi_n]- I_\ome [\Psi_n - \Psi_\infty] - I_\ome [\Psi_\infty] \to 0.
\ee
Suppose that $I_\ome [\Psi_\infty] >0 $. Then, by \eqref{claudio} and \eqref{liebbata},
\[
\lim_{n\to \infty} I_\ome [\Psi_n - \Psi_\infty]=   \lim_{n\to \infty}I_\ome [\Psi_n] - I_\ome [\Psi_\infty]
= - I_\ome [\Psi_\infty] <0 .
\]
Choose $\bar n$ such that $ I_\ome [\Psi_n - \Psi_\infty]<0$ for $n>\bar n$. Then by definition of
$d(\ome)$ we have
\be \label{partial1}
d (\ome) \leqslant \widetilde S [\Psi_n - \Psi_\infty ] , \qquad n>\bar n.
\ee
On the other hand, since $\Psi_\infty \neq 0$, by \eqref{breli3} one has
\[
\lim_{n\to \infty} \widetilde S [\Psi_n - \Psi_\infty]=   \lim_{n\to \infty}  \widetilde S [\Psi_n] 
-  \widetilde S  [\Psi_\infty]
=d(\ome)  -  \widetilde S  [\Psi_\infty] < d(\ome)
\]
and this contradicts \eqref{partial1}; so it must be $I_\ome [\Psi_\infty]\leqslant 0$.

\n
By definition $d(\ome) \leqslant \widetilde S  [\Psi_\infty]$. On the other hand, by the lower semicontinuity
of the norm under weak convergence we have
\[
\widetilde S  [\Psi_\infty] = \f{\mu}{2(\mu+1) } \| \Psi_\infty \|_{2\mu+2}^{2\mu +2} 
\leqslant \lim_{n\to \infty} \f{\mu}{2(\mu+1) } \| \Psi_n\|_{2\mu+2}^{2\mu +2} = d(\ome)
\]
which implies 
\[
\widetilde S  [\Psi_\infty] = d(\ome)
\]
and so $\Psi_\infty$ is an absolute minimum of $S_\ome$ constrained to Nehari manifold.
\end{dem2}


\section{Stationary States}\label{sezione5}
In this section we explicitly compute the stationary states of
$S_\ome$ and of $S_\ome^0$ and identify the minimum of the action.  We denote by  $[s]$  the integer part of $s$.

\begin{teo}[Stationary states of $S_\ome$]\mbox{ } \label{frullato}

\n
Let $\al<0$ and $\ome > \frac{\alpha^2}{N^2}$; then $S_\ome$ has $[(N-1)/2]+1  $ critical points $\Psi_{\ome,j},\  $  
with $ j=0,\dots ,[\frac{N-1}{2}] $, given, up to permutations of the
edges, by: 
\beq
\lf(\Psi_{\omega,j}\ri)_i(x) = 
\begin{cases}
\phi_s(x-a_j ) & i=1,\ldots ,j \\
\phi_s(x+a_j) & i=j+1, \ldots, N
\end{cases} \label{states1}
\eeq 
\beq a_j = \f{1}{\mu \sqrt{\ome}} \arctanh
\lf(\f{\al}{(2j-N)\sqrt{\ome}} \ri) \ .\label{states2} 
\eeq 
Moreover, for $-N\sqrt{\ome} < \al <\al^\ast$ the function $\Psi_{\ome , 0}$ is the ground state.
\end{teo}
\begin{dem}
A regularity argument shows that a constrained critical point of the action $S_\omega$ is in fact an element of the domain of the operator $H\ .$ We sketch the standard proof. Any such non vanishing critical point $\Psi$ satisfies $S^\prime (\Psi) = 0$, i.e.
\be \label{eins}
S_{\omega}^\prime (\Psi) \Phi \ = \ 0, \quad \forall \Phi \in {\mathcal E}.
\ee
Applying \eqref{eins} first to $\Phi$, then to $\Xi = - i \Phi$, and
summing the two expressions, we find
\be \label{zwei}
B (\Psi, \Phi) - \f{1}{2\mu+2} ( |\Psi |^{2 \mu} \Psi, \Phi) + \omega (\Psi, \Phi) \ = \ 0  
,
\ee
where $B$ is the bilinear form associated to the quadratic form $E^{lin}$.
So, from \eqref{zwei} the following estimate holds
\be\label{klmn}
| B (\Psi, \Phi) |  \ \leq \ C_\Psi \| \Phi \|, \qquad \forall \,
\Phi \in {\mathcal E}.
\ee
Notice that, choosing $\Phi$ among the functions vanishing in  a
neighborhood of zero, we conclude from \eqref{klmn} Riesz theorem and definition of weak derivative that every $\psi_i \in H^2 (\RE^+)$. Thus, for a generic $\Phi\in {\mathcal E}$ an integration by parts gives
\be \label{bigamma2}
2\ B (\Psi, \Phi) \  = -\sum_{i=1}^N( {\psi_i}^{\prime\prime}, \phi_i)_{L^2(\RE^+)}\  - {\phi}_1 (0)
\overline{ \left( \alpha\psi_1(0)-\sum_{i=1}^N {\psi_i}^{\prime}(0)\right)} .
\ee
So, from \eqref{klmn} and \eqref{bigamma2}, we conclude that $\Psi$ belongs to the domain  $\DD(H)$.  Moreover, the function
$
 H \Psi - \frac{1}{2\mu+2}| \Psi |^{2 \mu} \Psi + \omega \Psi
$ belongs to
$L^2 (\RE^+)$.




\noindent
Therefore $S'_{\omega}[\Psi]=0$ is equivalent to the following equation
\beq
\label{stat-eq}
H\Psi_{\omega} - |\Psi_{\omega}|^{2\mu} \Psi_{\omega} = -\ome \Psi_{\omega} \,\qquad \ome>0.
\eeq
Notice that $H$ acts locally as the Laplacian, thus on every edge we
must seek $L^2(\RE^+)$-solutions to the equation
\[
-\phi'' - |\phi|^{2\mu} \phi = -\ome \phi \,\qquad \ome>0.
\]
The most general $L^2 (\erre^+)$-solution is $ \phi(x) = \sigma \phi_s
(x- y) = \sigma \lf[ (\mu+1)\ome\ri]^{\f{1}{2\mu}} 
\sech^{\f{1}{\mu}} (\mu \sqrt{\ome} (x-y)) $
where $\sigma \in \ci$, $|\sigma|=1$ and $y\in\RE$. Therefore the components $(\Psi_{\omega})_i$
of a critical point $\Psi_{\omega}$  are given by
\begin{equation*}
\lf(\Psi_{\omega}\ri)_i(x)=\sigma_i \phi(x-y_i )\,.
\end{equation*}
In order to have a solution of \eqref{stat-eq} it is sufficient to impose boundary conditions \eqref{domdelta}
such that $\Psi_\ome \in \DD (H)$.  The
continuity condition in \eqref{domdelta} implies
$\sigma_1=\ldots=\sigma_N$ and $y_i=\ve_i a$ with $\ve_i=\pm1$ and
$a>0$. We can omit the dependence on $\sigma$ without losing generality.
Referring to the bell shape of the function $\phi_s$, we say that in the
i-th edge: there is a {\em bump} if $y_i>0$, that is, if $\ve_i=+1$;
there is a {\em tail} if $y_i<0$, that is, if $\ve_i=-1$.
Now we determine $\ve_i$ and $y_i$. The second boundary condition in
\eqref{domdelta} rewrites as
\beq 
\tanh (\mu \sqrt{\ome} a)\sum_{i=1}^N
\ve_i =\f{\al}{\sqrt{\ome}}\ .
\label{eq-a}
\eeq
Equation \eqref{eq-a} gives as a first constraint that $\sum_{i=1}^N\ve_i$ must have the same sign of $\al$. That is 
the critical point must have more tails than bumps. For every such a
configuration, or equivalently a choice of the set $\{\ve_i\}$, 
condition \eqref{eq-a} fixes uniquely $a$. We choose to index the
solutions by the number $j$ of bumps. Correspondingly one obtains a
unique solution to \eqref{eq-a} which we call $a_j$. 
In this way we arrive at \eqref{states1} and \eqref{states2}. For
instance, if $N=3$ then there are two stationary states, a three-tail state
and a two-tail/one-bump state (up to permutations of the edges). They are shown in figure \ref{figu1}. 

\mbox{}

\begin{figure}[h!] 
\centering
\includegraphics[scale=0.50]{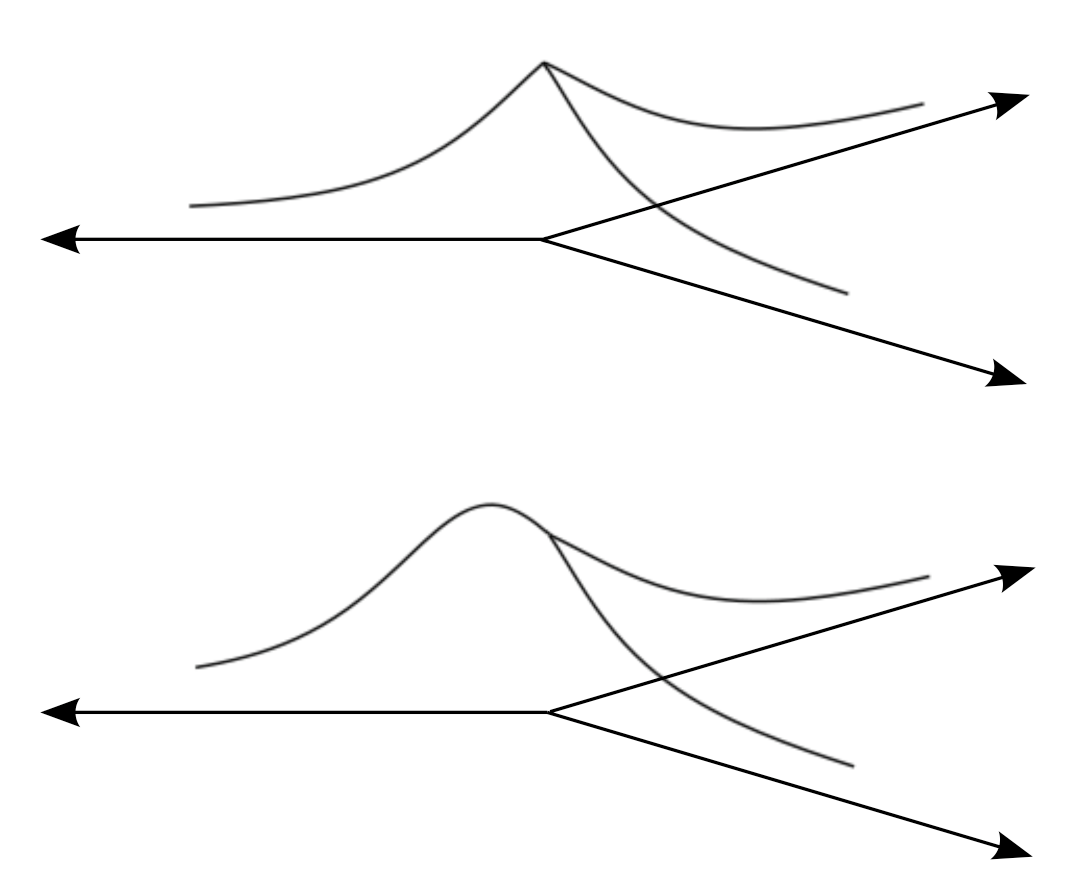}
\caption{Stationary states for $N=3$, $\alpha<0$\ .}
\label{figu1}
\end{figure}
\n
 
\noindent
To summarize, solutions to \eqref{stat-eq} 
are given by $\Psi_{\omega,j}$ with $j=0 ,\ldots ,[(N-1)/2]$.

\n
Notice that \eqref{eq-a} admits solutions iff the lower bound $\frac{{\al}^2}{N^2}<\omega\ $  holds true. We can explain this
fact for $\alpha <0$ noticing that \eqref{stat-eq}, for small $\Psi$ and neglecting
nonlinearity, is the eigenvalue equation for the linear part of the
Hamiltonian corresponding to energy $E=-\omega$; taking into account
the known fact that the linear graph Hamiltonian $H$ has the
ground state energy $-\frac{\alpha^2}{N^2}$, the lower bound means
that the nonlinear standing waves bifurcate from the vanishing
wavefunction at the ground state energy of the linear problem. 










\n
Now we prove that $\Psi_{\ome,0}$ is the ground state. Notice also
that $\Psi_{\ome,0}$ is uniquely defined since it is invariant under
permutations of the edges. We know that for $-N\sqrt{\ome} < \al
<\al^\ast$ a minimum of $S_\ome$ exists and therefore it is a critical
point. 
It is sufficient to prove that $S_\ome [ \Psi_{\ome, 0}] < S_\ome [ \Psi_{\ome, j}]$ for $j\neq 0$.
In fact we prove a stronger statement, that is, if $0 \leqslant j \leqslant [(N-1)/2] -1 $ then
\be \label{firb}
S_\ome [ \Psi_{\ome, j}] < S_\ome [ \Psi_{\ome, j+1}]\ .
\ee
Using \eqref{formula2}, equation \eqref{firb} is equivalent to
\begin{multline} \label{kuso}
j \int^1_{-\frac{|\al|}{(N-2j)\sqrt{\ome}   }} (1-t^2)^{\frac{1}{\mu} } dt 
+(N-j)\int^1_{\frac{|\al|}{(N-2j)\sqrt{\ome}   }} (1-t^2)^{\frac{1}{\mu} } dt < \\
(j+1) \int^1_{-\frac{|\al|}{(N-2j-2)\sqrt{\ome}   }} (1-t^2)^{\frac{1}{\mu} } dt 
+(N-j-1)\int^1_{\frac{|\al|}{(N-2j-2)\sqrt{\ome}   }} (1-t^2)^{\frac{1}{\mu} } dt.
\end{multline}
Let us define the constant $C$
\[
C = \int_0^1 (1-t^2)^{\frac{1}{\mu} } dt .
\]
It is convenient to rewrite the l.h.s. of \eqref{kuso} as
\begin{multline} \label{kuso2}
j \int^1_{-\frac{|\al|}{(N-2j)\sqrt{\ome}   }} (1-t^2)^{\frac{1}{\mu} } dt 
+(N-j)\int^1_{\frac{|\al|}{(N-2j)\sqrt{\ome}   }} (1-t^2)^{\frac{1}{\mu} } dt =
2j C + (N-2j) \int^1_{\frac{|\al|}{(N-2j)\sqrt{\ome}   }} (1-t^2)^{\frac{1}{\mu} } dt =\\
NC - (N-2j) \int_0^{\frac{|\al|}{(N-2j)\sqrt{\ome}   }} (1-t^2)^{\frac{1}{\mu} } dt.
\end{multline}
Repeating the same manipulations for the r.h.s., we see that \eqref{kuso2} is equivalent to
\[
(N-2j-2) \int_0^{\frac{|\al|}{(N-2j-2)\sqrt{\ome}   }} (1-t^2)^{\frac{1}{\mu} } dt <
(N-2j) \int_0^{\frac{|\al|}{(N-2j)\sqrt{\ome}   }} (1-t^2)^{\frac{1}{\mu} } dt .
\]
With a straightforward change of variables the last inequality becomes
\[ 
 \int_0^{\frac{|\al|}{\sqrt{\ome}   }} \lf[ 1-\lf( \f{t}{N-2j-2}\ri)^2\ri]^{\frac{1}{\mu} } dt <
 \int_0^{\frac{|\al|}{\sqrt{\ome}   }} \lf[1-\lf( \f{t}{N-2j}\ri)^2\ri]^{\frac{1}{\mu} } dt 
\]
which is manifestly true.
\end{dem}
\begin{remark} For $N>2$ and $j>0$ there exist excited states, but only for parameters $\omega>\frac{\alpha^2}{(N-2j)^2}$. So the picture is that for fixed $\alpha$ and increasing $\omega$ firstly the branch of ground state is born at $\omega>\frac{\alpha^2}{N^2}$ and then for sufficiently high $\omega$ the branches of higher excited states appear.
\end{remark}
\begin{remark} Even if in the present section we considered the case
  $\al \leqslant 0$, notice that the analysis of the previous theorem
  can be repeated also for  
$\al >0$ and one would find that the critical points are given again by \eqref{states1} and \eqref{states2} 
with $j=[N/2 + 1], \ldots , N$. This means that for a repulsive $\delta$ interaction at the vertex the stationary states have more bumps than tails.
\end{remark}
\par\noindent
We end this section with the characterization of stationary points of $S_\omega^0$
\begin{teo}[Critical Points of $S_\ome^0$] \mbox{ }
Let $\omega>0 $. If $N$ is odd, then there is a unique critical point of $S_\ome^0$ given by 
\be \label{milk}
\lf( \Psi_\ome^0 \ri)_i = \phi_s \qquad \qquad i=1,\ldots ,N
\ee
If $N$ is even then $S_\ome^0$ has a one parameter family of critical points given by:
\beq \label{shake}
\lf(\Psi_{\omega}^{0,a} \ri)_i(x) = 
\begin{cases}
\phi_s (x-a) & i=1,\ldots, N/2 \\
\phi_s (x+a) & i=N/2+1, \ldots ,N
\end{cases} \quad a\in\erre^+
\eeq
\end{teo}
\begin{dem}
\noindent
Repeating the argument in the proof of Theorem \ref{frullato}, we have to find the solutions of 
\begin{equation*}
\tanh (\mu \sqrt{\ome} a)\sum_{i=1}^N
\ve_i =0 .
\end{equation*}
If $N$ is odd, then there is a unique solution given by $a=0$ which
corresponds to the critical point \eqref{milk}, and such a solution
can be described as composed by
$N$ half solitons glued at the vertex.
On the contrary, if $N$ is even then there are infinite solutions: $a\in \erre^+$, $\ve_i = +1$ for $i=1,\ldots, N/2$, 
$\ve_i = -1$ for $N/2+1, \ldots ,N$ gives a solution to \eqref{shake}
which corresponds to $\Psi_{\omega}^{0,a}$. 
\end{dem}
\begin{remark}
If $N$ is even, then the graph can be
considered as a set of $N/2$ copies of the real line. With a Kirchhoff
boundary condition, one has continuity and derivability of the
wavefunction at the vertex, and the above solutions 
$\Psi^{0,a}_{\omega}$ can be interpreted as $N/2$ identical solitary waves on each real line translated by a quantity $a$. 
\end{remark}
\begin{remark} In the case $N=3$ and for a cubic nonlinearity it has been proved in \cite{ACFN3} that the energy $E$ at constant mass $M$ is not minimized on $\Psi_\ome^0$, which turns out to be a saddle point. In fact, the constrained energy is bounded from below but it has not an absolute minimum. We conjecture that the same phenomenon
happens here for the action.
\end{remark}

\section{Stability of Ground States}
In Section 4 we showed the existence of a profile
$\Psi_{\omega,0}$ (denoted there by $\Psi_{\infty}$) which minimizes
the action $S_\omega$ for the star graph with attractive delta
boundary conditions at the vertex. This minimizer is the ground state
of the problem if the strength $\alpha$ of the point interaction at
the vertex is sufficiently large. In Section 5 we provided the
explicit expression of stationary states $\Psi_{\omega,j}$, and in
particular of the ground state $\Psi_{\omega,0}$. In correspondence to
the ground state (and to every stationary state) one has a standing
wave of the form $\Psi_{\omega}\ e^{i\omega t}$ which solves the NLS
on the graph. In this section we study the stability of such a
standing wave. Being a time-dependent solution and not an equilibrium
point of the autonomous equation \eqref{diffform}, stability has to be
intended as {\it orbital stability}. This means Lyapunov stability up to
symmetries of the equation, which in this case are related to the gauge
$U(1)$ invariance of the Hamiltonian of the problem. To be precise, we
recall that the orbit of $\Psi_{\omega}$ is defined as ${\mathscr
  O}(\Psi_{\omega})=\{e^{i\theta}\Psi_{\omega},\ \theta\in
\erre\}$. \par\noindent The state $\Psi_{\omega}$ is orbitally stable
if for every $\epsilon>0$ there exists $\delta>0$ such that
\begin{equation*}
d(\Psi(0), {\mathscr O}(\Psi_{\omega}))<\delta \quad \Rightarrow \quad
d(\Psi(t), {\mathscr O}(\Psi_{\omega}))<\epsilon \quad\quad \forall t
> 0 
\end{equation*}
where $\Psi(t)$ is the solution to \eqref{intform1} with initial data $\Psi_0$,
\begin{equation*}
d(\Psi, {\mathscr O}(\Psi_{\omega}))=\inf_{\Phi\in{\mathscr O}(\Psi_{\omega})}\|\Psi-\Phi\|_{\EE},
\end{equation*}
and the norm $\|\cdot\|_{\EE}$ is the energy norm, given in our case by $H^1$ norm in ${\mathcal E}$.\par\noindent
A stationary state is unstable if it is not stable.\par\noindent
Orbital stability of solitary (not necessarily standing) solutions to
nonlinear Schr\"odinger equations is a well developed subject, studied
in several classical papers. Two main techniques have been developed
to establish orbital stability of solitary waves: the method of
Cazenave and Lions based on Concentration Compactness (\cite{[Caz],
  [CL]}), and the method of constrained linearization pioneered by
Benjamin in the case of KdV equation and studied more systematically
by Weinstein and Grillakis-Shatah-Strauss \cite{[W2],[W3], [GSS1],
  [GSS2]}. 
In \cite{[ACFN4]} we studied the problem of the   minimization of the energy at constant mass through a suitable adaptation of concentration compactness method to the case of star graphs, while
here we refer to the Weinstein and Grillakis-Shatah-Strauss method
which is especially suited for treating stability of equilibria of Hamiltonian systems with symmetry. Some preparation is needed to cast
our problem in this framework. As in the scalar case, the NLS on a
graph turns out to be a Hamiltonian system on the real Hilbert space
of the couples of real and imaginary part of the wavefunction. We pose
$\Psi=U\ + i\ V \equiv (U,V)$, where $U=(u_1,\dots,u_N)^T$ and
$V=(v_1,\dots,v_N)^T$. So we identify $L^2(\GG)=L^2(\GG, \ci)$ with
$L^2(\GG, \erre)\oplus L^2(\GG, \erre):=L^2_{\erre}(\GG)\ .$
Analogously one can define the spaces
$L^{p}_{\erre}(\GG)$. Correspondingly,  
$L^2(\GG)$ can be given the structure of a real Hilbert space taking
as its scalar product the real part of the usual complex one: 
\[
((U_1,V_1)^T,(U_2,V_2)^T )_{L^2_{\erre}(\GG)} = {\rm Re} (\Psi_1, \Psi_2 )_{L^2(\GG)}\ .
\]
Furthermore, $L^2(\GG)$ is also a symplectic manifold when endowed
with the symplectic form (coinciding with the imaginary part of the
complex scalar product) 

\[
\Omega((U_1,V_1),(U_2,V_2))={\rm Im}( \Psi_1,  \Psi_2)_{L^2(\GG)} =
\sum_{i=1}^N \int_{\erre^+} ((v_2)_i (u_1)_i - (v_1)_i (u_2)_i)
dx\\ . 
\]
The same symplectic structure is inherited by the energy space $\EE$.
Moreover, multiplication by the imaginary unit $i$ is equivalent to
acting by the matrix $-\mathcal J \in {\rm Mat}(\erre, 2N\times 2N)$,
where  
\begin{equation*}
\mathcal J = \left(
  \begin{array}{cc}
    0 & I \\
    -I & 0 \\
  \end{array}
\right)\ ,
\end{equation*}
and the blocks $0$ and $I$ are the zero and unit matrices in ${\rm Mat}(\erre, N\times N)$.\par\noindent
Note that if $\Psi\in \DD(H)$, then the real vectors $U$ and $V$
satisfy the same boundary conditions as $\Psi$; we will say, with a
slight abuse, that they belong to $\DD(H)$.  With these premises, the
nonlinear Schr\"odinger equation for $\Psi$ is equivalent to the
canonical system  
 \be
\label{hamilt}
\frac{d}{dt} 
\begin{pmatrix}
U \\ V
\end{pmatrix}={\mathcal J} E'[U,V]
.
\ee
where the Hamiltonian $E$ becomes
\[
E(U,V)=\frac{1}{2}\|(U',V')\|_{L^2_{\erre}(\GG)}^2  -
\frac{1}{2\mu+2}\|(U,V)\|^{2\mu+2}_{L^{2\mu+2}_{\erre}(\GG)} +
\frac{\alpha}{2}(|u_1(0)|^2+|v_1(0)|^2) 
\]
or explicitly
\[
\begin{aligned}
& \frac{1}{2}\sum_{k=1}^N\bigg[\int_0^{+\infty}(|u_k'|^2 + |v_k'|^2)\ dx \bigg]+ \frac{\alpha}{2}(|u_1(0)|^2+|v_1(0)|^2)\\
& -\frac{1}{2\mu+2}\sum_{k=1}^N\bigg[\int_0^{\infty}(|u_k|^2+|v_k|^2)^{\mu+1}dx\bigg]\equiv E[u,v]\ ,
\end{aligned}
\]
\par\noindent
and the derivative $E'$ is given by
$$
E'((U,V))[(H,Z)]={\frac{d}{d\epsilon}} \{E((U,V)+\epsilon (H,Z))\}_{\epsilon=0}.
$$
Linearization of the Hamiltonian system \eqref{hamilt} around the stationary state is achieved by substituting
\[
(\Psi_t)_j=((\Psi_{\omega,0})_j + h_j +i z_j)e^{i\omega t}
\]
and neglecting higher order terms than linear in \eqref{hamilt}. The real vector functions $H$ and $Z$ satisfy
 \begin{equation*}
\frac{d}{dt} 
\begin{pmatrix}
H \\ Z
\end{pmatrix}={\mathcal J} {\mathcal L}
\begin{pmatrix}
H \\ Z
\end{pmatrix}\ ,
\end{equation*}
where ${\mathcal L}$ is the unique s.a. operator associated to the
symmetric and lower bounded quadratic form
$S_{\omega}''(\Psi_{\omega,0})$, i.e. the second derivative of the action at the
ground state. Indeed, the second derivative is defined through the formula 
 $${\mathcal L}((H_1,Z_1),(H_2,Z_2))=
 S_{\omega}''(\Psi_{\omega,0})((H_1,Z_1),(H_2,Z_2))=\frac{\partial^2}{{\partial\epsilon}{\partial \lambda }}
\{{S_{\omega}}(\Psi_{\omega}+\epsilon (H_1,Z_1) +\lambda (H_2,Z_2)\}_{\epsilon=0, \lambda=0}\ . $$
An easy computation shows that 
${\mathcal L}={\rm diag}({\mathcal L}_-,{\mathcal L}_+)$ and the
matrix operators ${\mathcal L}_-$ and ${\mathcal L}_+$ are given by
(here the summation convention is used) 
\begin{equation*}
\begin{aligned}
\left({\mathcal L}_+\right)_{i,k} &=\left(-\frac{d^2}{dx^2} + \omega - |(\Psi_{\omega,0})_k|^{2\mu}\right){\delta}_{i,k} \\ 
\left({\mathcal L}_-\right)_{i,k} &=\left(-\frac{d^2}{dx^2} +\omega -
(2\mu+1)|(\Psi_{\omega,0})_k|^{2\mu}\right){\delta}_{i,k} \ . 
\end{aligned}
\end{equation*} 
The operators ${\mathcal L}_-$ and ${\mathcal L}_+$  act on the real
vector functions $H$ and $Z$ belonging to ${\mathcal D}(H)={\mathcal
  D}({\mathcal L}_\pm)$. Notice that to simplify notation from here on
we suppress the dependence of operators ${\mathcal L}_\pm$ on the
ground state $\Psi_{\omega,0}$. \par\noindent  
Precise conditions to have orbital stability (and instability) for general Hamiltonian
systems and in particular for systems of NLS equations, are given in
the already quoted papers of Weinstein and
Grillakis-Shatah-Strauss. They can be reduced to the validity of three
conditions, called Assumptions I, II and III in \cite{[GSS1]} and
\cite{[GSS2]} and the verification of a further convexity condition on
the function $d(\omega)=S_{\omega}(\Psi_\omega)$  introduced in
Section 4 and called in the physical literature the Vakhitov-Kolokolov
condition.\par\noindent 
Assumption I is the  well-posedness, proved in Section 2. Assumption II is 
the existence of a regular branch of standing solutions of the
stationary equation, proved for our model in Section 4 and 5, where
the regular family of standing waves $(\frac{\alpha^2}{N^2},
+\infty)\ni\omega \mapsto\Psi_{\omega,0}$ is explicitly
constructed. Assumption III concerns spectral properties of
linearization $E''(\Psi_{\omega,0})=({\mathcal L}_-, {\mathcal L}_+)$
around the ground state. The spectral conditions are stated and proved
in the following proposition. 
\begin{prop} \label{spectral}\par\noindent The operators ${\mathcal
    L}_-$ and ${\mathcal L}_+$ are selfadjoint. Moreover:\par\noindent 
$i_1$) $\ker {\mathcal L}_+ = \{\Psi_{\omega,0}\}$ and
the rest of the spectrum is positive; \par\noindent
$i_2$)  $\ker{\mathcal L}_- = \{0\}$; \par\noindent
$i_3$) $n({\mathcal L}_-)=1$, 
where $n(A)$  is the number of negative eigenvalues of the operator $A$, i.e. its Morse index. 
\end{prop}
\begin{dem}
We begin to remark that operators ${\mathcal L}_-$ and ${\mathcal
  L}_+$ are selfadjoint on ${\mathcal D}({\mathcal L}_\pm)={\mathcal
  D}(H)$, due to the fact that the components of the ground state
$(\Psi_{\omega,0})_k$ are continuous and strongly decaying at
infinity, and as such they constitute in the matrix operators
${\mathcal L_{\pm}}$ a relatively compact perturbation of $H
+\omega$. For the same reason, by Weyl's theorem the absolutely
continuous spectrum of ${\mathcal L}_-$ and ${\mathcal L}_+$ coincides
with the essential spectrum of  $H +\omega$, i.e. $[\omega, +\infty)$,
  and the discrete spectrum is composed at most of a finite number of
  eigenvalues. 
Let us consider the kernel of ${\mathcal L}_+$. This surely contains
$\Psi_{\omega,0}$. Indeed,  the equation ${\mathcal
  L}_+{\Psi_{\omega,0}}=0$ coincides with the stationary equation
satisfied by $\Psi_{\omega,0}$. Let us show that there are not other
elements in the kernel. An integration by parts allows to rewrite the
quadratic form of  ${\mathcal L}_+$, for any element $V\in {\mathcal
  D}({\mathcal L}_+)$, as follows 
\[
\begin{aligned}
 ({\mathcal L}_+ V, V)_{L^2_{\erre}(\GG)}= &\sum_{k=1}^N\int_0^{+\infty} ((\Psi_{\omega,0})_k)^2 |\frac{d}{dx}(\frac{v_k}{(\Psi_{\omega,0})_k})|^2 dx\\+ & \sum_{k=1}^N\left(v_k(0)v'_k(0)-|v_k(0)|^2\frac{({\Psi}'_{\omega,0})_k(0)}{({\Psi}_{\omega,0})_k(0)} \right)
\end{aligned}
\]
and the last term is vanishing due to the $\delta$ boundary conditions at the vertex, continuity and $\sum_{k=1}^N v'_k(0)=\alpha v_1(0)$. \par\noindent
So $({\mathcal L}_+ v,
v)_{L^2_{\erre}(\GG)} > 0 $ for every $v\in {\mathcal D}({\mathcal L}_+)$ not coinciding with $\Psi_{\omega,0}$, which is the only eigenvector with eigenvalue $0$ of the operator ${\mathcal L}_+$. This proves statement $i_1$. \par\noindent Concerning statement $i_2$, it is sufficient to consider the equation ${\mathcal L}_- u=0$. This is written, in components, as 
\be\label{eqL-}
-\frac{d^2}{dx^2}u_k +\omega u_k- (2\mu+1)|(\Psi_{\omega,0})_k|^{2\mu} u_k = 0\ \ \ k=1,\dots, N\ ,
\ee
where $u_1(0)=u_2(0)=\dots =u_N(0)$ and $\sum_{k=1}^N u'_k(0)=\alpha
u_1(0)$ due to boundary conditions. The general theory of second order
differential equations gives for the previous equation a solution
which is a linear combination of asymptotically exponential
fundamental solutions; for $x\to +\infty$ only one of them is in
$L^2$. Now notice that the function $\tilde u_k:(0,+\infty)\mapsto
\erre\ \ {\rm s.t.}\ \ \tilde u_k(x)=\frac{d}{d
  x}(\Psi_{\omega,0})_k(x),\ \ k=1,\dots,N$ satisfies equation
\eqref{eqL-} in $(0,+\infty)$ and  decays at infinity. So every
solution in $u\in L^2(0,+\infty)$ of \eqref{eqL-} is a multiple of
such a function: $u_k (x)=c_k \tilde u_k (x), c_k\in \erre$. To
conclude, a direct calculation shows that the only real constants
compatible with the boundary conditions on $u$ is $c_k=0$, $k=1,\dots
, N$. This proves statement $i_2\ .$\par\noindent 
Let us consider finally the Morse index of the ${\mathcal L}_-$
operator. It is immediate that $n({\mathcal L}_-)\geq 1\ .$  Indeed,
let us consider the quadratic form for ${\mathcal L}_-$ evaluated on
the ground state:  
\[
({\mathcal L}_-\Psi_{\omega,0},\Psi_{\omega,0})_{L^2_{\erre}(\GG)} =
({\mathcal L}_+\Psi_{\omega,0},\Psi_{\omega,0})_{L^2_{\erre}(\GG)}
-2\mu (
|\Psi_{\omega,0}|^{2\mu}\Psi_{\omega,0},\Psi_{\omega,0})_{L^2_{\erre}(\GG)}
= 0 - 2 ||\Psi_{\omega,0} ||^{2\mu+2}_{2\mu+2} <0\ . 
\]
So the s.a. linear operator ${\mathcal L}_-$ has a negative vector,
so it surely admits at least negative eigenvalue. Let
us prove that it has  a single negative eigenvalue only. This is a
consequence of the variational properties of $\Psi_{\omega,0}$. In
fact
 $\Psi_{\omega,0}$ is a minimum point of the action $S_{\omega}$
on the codimension one constraint $I_{\omega}=0$. This minimization
property entails that $S''(\Psi_{\omega,0})$ is positive definite on
the tangent space at $\Psi_{\omega,0}$ of the constraint
manifold. Being the constraint a manifold of codimension one,
$S''(\Psi_{\omega,0})$ admits at most one negative eigenvalue and the
same is true for its only possibly negative diagonal component
${\mathcal L}_-$. See Appendix B in \cite{FGJS1} for the detailed 
argument. 
\end{dem}

The last property needed to show orbital stability of the ground state
is the so called slope condition, or Vakhitov-Kolokolov
condition. This coincides with the convexity of the function
$d(\omega)$, or more explicitly it means that on the branch of
stationary solutions $\{\Psi_{\omega,0}\}$ parametrized by $\omega$,
one has  
$$d''(\omega)=\frac{d^2}{d\omega^2} S_\ome(\Psi_{\omega})=\frac{d^2}{d\omega^2} (E(\Psi_\omega)+\omega M(\Psi_\omega))=\frac{d}{d\omega}||\Psi_{\omega,0}||^2> 0\ .$$
\par\noindent
In fact, a direct calculation making use of the formulas in the appendix (and which is possible in this model due to the explicitly known form of  $\Psi_{\omega,0}$), gives
\be\label{VK}
\frac{d}{d\omega}||\Psi_{\omega,0}||^2_2=
C\bigg[(\frac{1}{\mu}-\frac{1}{2})\int^1_{\frac{|\alpha|}{N\sqrt{\omega}}} (1-t^2)^{\frac{1}{\mu}-1}\ dt +\frac{|\alpha|}
{2N\sqrt{\omega}}\left(1-\frac{|\alpha|^2}{N^2\omega}\right)^{\frac{1}{\mu}-1}\bigg]
\ee 
with
$C=C(N,\mu,\omega)=N\frac{(\mu+1)^{\frac{1}{\mu}}}{\mu}\omega^{\frac{1}{\mu}-\frac{3}{2}}$.\par\noindent
Now, the r.h.s of \eqref{VK} is positive thanks to the lower bound
on $\omega > \frac{\alpha^2}{N^2}$. 
The Vakhitov-Kolokolov condition with the spectral properties proved in proposition \ref{spectral}, thanks to the Weinstein and Grillakis-Shatah-Strauss theory constitute the proof of Theorem \ref{orbitalstab} stated in the introduction.
\par\noindent 
\begin{remark} Note the following facts. \par\noindent
The theorem gives orbital stability of the ground state also for the critical nonlinearity $\mu=2\ .$\par\noindent 
From formula \eqref{VK} it follows that for supercritical nonlinearities $\mu>2$ there exists $\omega^*$ such that $\Psi_{\omega,0}\ $ is orbitally stable for 
$\omega\in (\frac{{\alpha}^2}{N^2},\omega^*)\ .$ In \cite{[GSS1]} it
is shown that if Assumptions I, II, III are satisfied and
$d''(\omega)<0$, then the standing wave corresponding to $\omega$ is
orbitally unstable. Again from formula \eqref{VK} we see that for
$\omega>\omega^*$ the ground state $\Psi_{\omega,0}$ is orbitally
unstable. The case $\omega=\omega^*$ where $d''(\omega)=0$ is
undecided. 
\end{remark}

\appendix
\section{Rearrangements} \label{rearra}

\n
For a given function $\Phi:\GG \to \ci$ we introduce the rearranged function $\Phi^*:\GG \to \erre$. The function $\Phi^*$ is  positive, symmetric, non increasing and   is constructed in a way such that it is equimeasurable w.r.t. $\Phi$, that is, the level sets of $|\Phi|$ and $\Phi^*$ have the same measure. This is sufficient to prove that all the $L^p(\GG)$ norms are conserved by the rearrangement. The comparison of the kinetic energy of $\Phi$ and $\Phi^*$ is more delicate. On the real line  the P\'olya-Szeg\H o inequality shows that the kinetic energy does not increase. This is no longer true for a star graph where a constant $N/2$ appears, see Theorem \ref{polya} below.

\n
Given $\Phi: \GG \to \ci$, we introduce $\la (s)$ and $\mu (s)$ defined by 
\[
\la (s) = | \{ |\Phi| \geqslant s\} |  \qquad \qquad
\mu (s) = | \{ |\Phi| > s\} |\,.
\]
Now we define the symmetric rearrangement of $\Phi$.
\begin{dfn} \label{defiriar}
Define $g: \erre^+ \to \erre^+$ as
\[
g(t) = \sup \{ s | \, \la(s) > N t\}\,,
\]
then we put $\Phi^* = (\phi^*_1, ... , \phi^*_N)$ with 
\[
\phi^*_j(x) = g (x)\,\qquad \qquad j=1,\dots,N.
\]
\end{dfn}

\n 
The main properties of $\Phi^{\ast}$ are the following:
\begin{prop} \label{lp}
Let $\Phi\in L^p(\GG)\ .$
The symmetric rearrangement $\Phi^{\ast}$ is  positive, symmetric and non increasing. Moreover $\| \Phi^{\ast}\|_{p  } =\| \Phi\|_{p  }$.
\end{prop}
\begin{dem}
\n
By construction  $\Phi^\ast$ is positive and symmetric. Since $\la$ is
non increasing, $\Phi^\ast$ is non increasing too.

\n
Now we prove the invariance of the $L^p$ norm.
First we prove that if $ 0 \leqslant N t \leqslant \mu (0) $ then  
\be \label{micio}
\mu (g(t) ) \leqslant N t \leqslant \la (g(t)) .
\ee
Let $s'< g(t) $ then, by definition of $g$, $\la (s')> Nt$. Since the latter estimate holds for every 
such $s'$, taking the supremum we have $\la(g(t)) \geqslant Nt$ which
is the second half of \eqref{micio}.

\n
Now choose $s' > Nt$. Then $\mu (s') \leqslant \la (s') \leqslant Nt$ and taking the infimum over $s'$ the proof
of \eqref{micio} is complete. 

\n
A key property of the symmetric rearrangement is the equimisurability, that is:
\be
\label{equi}
|\{\Phi^{\ast} \geqslant s \} | = |\{ |\Phi| \geqslant s \}| \,.
\ee
\n
Notice that the set $\{ |\Phi|\geqslant t \}$ can be rephrased as the
union of
disjoint sets
\[
\{ |\Phi|\geqslant t \}= \{ |\phi_1| \geqslant t \}\cup \ldots \cup \{ |\phi_N|\geqslant t \} ,
\]
where $\{ |\phi_j| \geqslant t \}$ is to be understood as a subset of the $j$-th edge. Therefore we have
\[
|\{ |\Phi|\geqslant t \}| = \sum_{i=1}^N  | \{ |\phi_i| \geqslant t \} |.
\]
Fix $s$ and define
\begin{equation*}
t_0 = \sup \{ t | g(t) \geqslant s\}\,.
\end{equation*}
If $g(t)>s$ then $Nt \leqslant \la ( g(t) ) \leqslant \la (s) $ by \eqref{micio}. Taking the supremum
over $t$ we get $N t_0 \leqslant \la (s)$.

\n
Assume by absurd that $N t_0< \la (s) $. Take $t$ such that $N t_0 < Nt <\la (s)$ then 
$g(t) \geqslant s$ and the contradiction with the definition of $t_0$ is reached. Since
$N t_0 = |\{\Phi^{\ast} \geqslant s \} |$ equality  \eqref{equi} is proved.

\n
By the layer cake representation (see \cite{LL01}) we can prove the invariance
of $L^p$ norm under rearrangements. We start from
\[
\| \Phi\|_{p  }^p \equiv \sum_{j=1}^N \int_{\erre^+} |\phi_j (y)|^p dy = 
p\int_0^{+\infty}  s^{p-1}|\{ |\Phi|\geqslant s \}|\, ds
\]
and using \eqref{equi} we obtain
\[
\| \Phi \|^p_{p  } = 
p\int_0^{+\infty}  s^{p-1} | \{ \Phi^\ast \geqslant s \}| \, ds 
= \| \Phi^\ast \|^p_{p  }
\]
which is the desired identity.
\end{dem}

\n
Notice that if $\{|\Phi|=t\}$ has non-zero measure for some $t$ then
$\la$ has a jump, while if $|\Phi|$ has a jump then $\la$ has a
constant part. Moreover if $\la$ has a constant part then $g$ has a
jump and if $\la$ has a jump then $g$ has a constant part.  

\n
Notice also that if $|\Phi|$ is continuous and $\{|\Phi|=t\}$ has zero measure for $t$ then $g$ is the inverse of $\la$ up to scaling.





\n
Now we turn our attention to the P\'olya-Szeg\H o inequality and
prove it  by elementary methods.
We mainly follow \cite{hilden} while some technical results are taken from \cite{crucco}.

\n
From now on, we restrict ourselves  to real and positive $\Phi$. Later we shall extend the result to the general case.
First we gather some preliminary results in Lemmas \ref{cont} and \ref{lemmolo}.
Then we establish the required estimate for a class of regular
functions (Lemma \ref{linearzero} and Lemma \ref{linear}). Finally we
give the main theorem the proof of which relies on a careful
decomposition of the kinetic energy and on a limiting argument. 

\n
Moreover, since functions in $\EE$ are continuous, for our purposes in the following
we always assume that $\Phi$ is continuous without losing generality. 
\begin{lem} \label{cont}
Assume that $\Phi:\GG\to \erre^+$ is continuous and $\Phi\in L^p(\GG)$
then $\Phi^\ast$ is continuous and $\Phi^* \in L^p(\GG)$. 
\end{lem}
\begin{dem}
Due to proposition \ref{lp}, we have to discuss the continuity part only. 
Since $\Phi$ is continuous, $\la$ is  strictly decreasing and may have at most a countable number
of discontinuity of first kind. Therefore $\la$ is locally continuous away from
discontinuities and $g$ is locally (up to an irrelevant scaling) the
inverse function.  Then $g$ is locally continuous by the inverse
function theorem. At the points where $\la$ has a discontinuity it is
easy to check, using definition \ref{defiriar}, that $g$ has a 
constant part joining the non constant branches. So that  $g$ is
globally continuous. See \cite{hilden} for more details. 
\end{dem}

\n
Reasoning as in \cite{hilden} we get the following:
\begin{lem} \label{lemmolo}
Let $\Phi_n, \Phi:\GG\to \erre^+$ and $\Phi_n, \Phi \in L^p (\GG)$.
Suppose that $\|\Phi_n - \Phi\|_{p  } \to 0 $, then
\begin{align*}
&\mu (s) \leqslant \liminf_n \la_n (s) \leqslant \limsup_n \la_n (s) \leqslant \la (s)\,, \\
& g(0) \leqslant \liminf_n g_n (0)\,.
\end{align*}
\end{lem}

\n
We introduce the following class of regular functions.
\begin{dfn}
Let ${\mathcal PL}$ be the set of functions $\Phi:\GG \to \erre^+$ such that: $\Phi$ is continuous,  compactly supported and, for any $j=1,...,N$,  there
exists a finite number of compact intervals $I_{j,n}$  such that $\supp \phi_j = \bigcup_{n} I_{j,n}$ and $\Phi$ restricted to  $I_{j,n}$ is affine.
\end{dfn}
\n
This class of piecewise linear functions is dense in $\EE$.
\begin{lem} \label{linearzero}
Let $\Phi\in {\mathcal PL}$. There exist two open sets $O_1$ and $O_2$ such that $\Phi^\ast$ is constant
on $O_1$ and $\Phi^\ast$ is differentiable on $O_2$ with $|{\Phi^{\ast }}'|>0$. Moreover, $\GG \setminus ( O_1
\cup O_2 )$ consists of finitely many points.
\end{lem}
\begin{dem}
Let $0=a_0<a_1< \ldots < a_m$ be the values assumed by $\Phi$ at the boundary of all $I_{j,n}$. If the set 
$\cup_i \{ \Phi= a_i \}$ has positive measure then $\Phi'=0$ a.e. on such a set. In the same way
${\Phi^{\ast}} '=0$ a.e. on the set $\cup_i \{ \Phi^\ast = a_i \}$. We define $O_1$ to be the interior part of
$\cup_i \{ \Phi^\ast = a_i \}$
and $O_2= \GG \setminus \cup_i \{ \Phi^\ast = a_i \}$. By construction
$\GG \setminus ( O_1 \cup O_2 )$ consists of finitely many points. 

\n
We have to show that $\Phi^\ast$ is differentiable on $O_2$ and $|{\Phi^{\ast}}'|>0$.
We introduce the notation $D_i=  \{ a_{i-1}<\Phi <a_i\}$ and $D_i^\ast=  \{ a_{i-1}< \Phi^\ast <a_i\}$.
Each set $D_i$ is decomposed first into the components on each edge, that is, $D_i= \cup_j D_i^j$.
Then we further decompose each $D_i^j$ 
into a finite union of open intervals $D_{i,k}^j$, $k=0,1, \ldots , n=n(i,j)$,
such that $\Phi$ restricted to $D_{i,k}^j$ is affine and non constant (see Figure \ref{figu3}).
Let us fix $s$ such that $ a_{i-1}< s <a_i$. Then the equation $\Phi=s$ has a solution $y_{i,k}^j (s) \in D_{i,k}^j$ for each $k=1,\ldots
, n$. We enumerate the sets $D_{i,k}^j$ in $k$ for $i,j$ fixed in an increasing way w.r.t to the distance 
from the vertex such that $y_{i,1}^j<y_{i,2}^j< \ldots < y_{i,n}^j$.
We put $n(i,j)=0$ if $D_i^j= \emptyset$ and no $y_{i,k}^j$ is defined for that values of $i$ and $j$.

\mbox{}

\begin{figure}[h!] 
\centering
\includegraphics[scale=0.70]{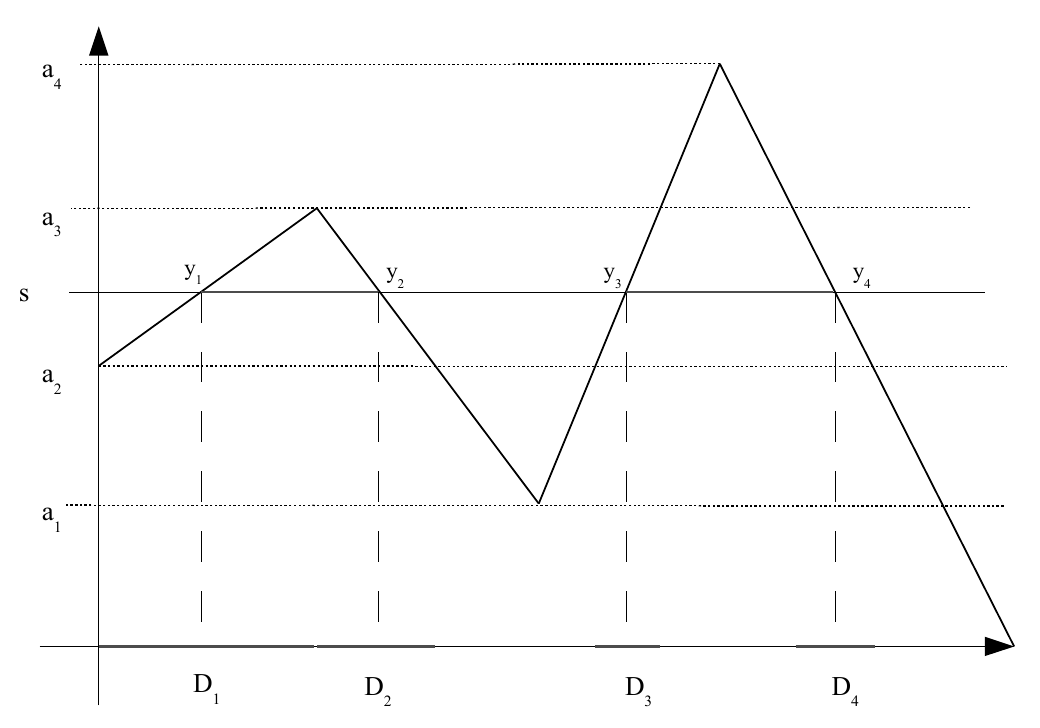}
\caption{Definition and enumeration of $D_{i,k}^j$}
\label{figu3}
\end{figure}
\n
We introduce also $ \tilde D_i^\ast$ defined as the projection of $D_i^\ast$ on the first edge.
Let $y^\ast (s)\in \tilde D_i^\ast$ be the solution to $g = s$.
We can express the measure of the level set $\{ \Phi >s\}$ by means of the $y_{i,k}^j$.
One has 
\[
|\{ \Phi >s\} | = \sum_{j=1}^N y_{i,n}^j - y_{i,n-1}^j + y_{i,n-2}^j -y_{i,n-3}^j\ldots =
\sum_{j=1}^N \sum_{k=1}^n (-1)^{n+k} y_{i,k}^j.
\]
Therefore by \eqref{equi} one has
\be \label{dragon}
|\{\Phi^\ast > s \}|= N |\{ g>s \} |= 
N y^\ast  = \sum_{j=1}^N \sum_{k=1}^n (-1)^{n+k} y_{i,k}^j.
\ee
On each set $D_{i,k}^j$ the derivative $\Phi'$ does not vanish by construction which implies that 
the function $ y_{i,k}^j (s) $ are differentiable w.r.t $s$ and
\[
\Phi' = \lf( \frac{d y_{i,k}^j}{ds} \ri)^{-1} \quad \text{on } D_{i,k}^j
\]
by the inverse function theorem. The function $y^\ast (s)$ is differentiable by equation \eqref{dragon}.
It is also invertible since $s$ is away from values of $\Phi$ corresponding to level sets with non zero measure. 
Therefore for such values of $s$ the function $g$ is strictly decreasing.
By the inverse function theorem $y^\ast$ is invertible and
\[
g'   = \lf( \frac{d y^\ast }{ds} \ri)^{-1} \neq 0.
\]
\end{dem}
Let $L$ be the Lipschitz constant of $\Phi$. Adapting the reasoning in
\cite{crucco}, one can prove  the following estimate:
\[
\lf| \frac{d y^\ast }{ds} \ri|\geqslant \f{1}{L}\,.
\]
This estimate provides an upper bound on $g'$. 
Notice that the proof actually shows that $\Phi^\ast \in {\mathcal PL}$ since it says that ${y^{\ast}} '$ and
therefore $g'$ is locally constant on each $\tilde D_i^\ast $.
If $\Phi$ is smooth, say $C^k$, then the same property holds on $O_2$ for $\Phi^\ast$ by the inverse function theorem.
\begin{prop} \label{linear}
Let $\Phi\in {\mathcal PL}$. Then,
\be \label{tripoli}
\|{\Phi^{\ast}} '\| \leqslant  \f{N}{2}    \|\Phi'\|  \,.
\ee
\end{prop}
\begin{dem}
We shall use the notation of the previous lemma.
First we consider the r.h.s. of  \eqref{tripoli}. We can restrict the
integral to the region where $\Phi$ is not constant and  change the integration variable.
\begin{equation*}
  \|\Phi'\|^2 
 = \sum_{i=1}^m \sum_{j=1}^N \sum_{k=1}^n \int_{D_{i,k}^j} |\Phi'|^2  = \sum_{i=1}^m \sum_{j=1}^N \sum_{k=1}^n \int_{a_{i-1}}^{a_i}  \lf( \f{dy_{i,k}^j}{ds} \ri)^{-2} 
 \lf| \f{dy_{i,k}^j}{ds} \ri| ds.
\end{equation*}
We can repeat the same operation for the l.h.s. of \eqref{tripoli}
\begin{equation*}
\|{\Phi^{ \ast}}'\|^2 =  N \int_{\erre^+} |g'|^2 =  N \sum_{i=1}^m \int_{a_{i-1}}^{a_i}  \lf( \f{d{y^{\ast}}'}{ds} \ri)^{-2} 
 \lf| \f{dy^{\ast }}{ds} \ri| ds.
\end{equation*}
Now the conclusion follows by using \eqref{dragon} and the convexity properties of the square function. 

\n
First notice that
\[
\lf|  \f{dy^{\ast }}{ds} \ri| = \f 1 N
\lf|  \sum_{j=1}^N \sum_{k=1}^n (-1)^{n+k} \f{d y_{i,k}^j}{ds} \ri| = \f 1 N
 \sum_{j=1}^N \sum_{k=1}^n  \lf| \f{d y_{i,k}^j}{ds} \ri|.
\]
The restriction of $\Phi$ to an edge has a seesaw
behavior and $\f{d y_{i,k}^j}{ds}$ has an alternating behavior in $k$.

\n
Therefore in order to prove \eqref{tripoli} it is sufficient to show that
\[
  \f{N^2}{4} \sum_{j=1}^N \sum_{k=1}^n  \lf( \frac{1}{ \lf| \f{d y_{i,k}^j}{ds} \ri|} \ri)^2  \lf| \f{d y_{i,k}^j}{ds} \ri|
\geqslant 
\lf(   \sum_{j=1}^N \sum_{k=1}^n  \frac{N}{ \lf| \f{d y_{i,k}^j}{ds} \ri|} \ri)^2 
 \sum_{l=1}^N \sum_{h=1}^n \lf| \f{d y_{i,h}^{l}}{ds} \ri|
\]
which is equivalent to
\be \label{atsui}
 \f{1}{4} \sum_{j=1}^N \sum_{k=1}^n  \lf( \frac{1}{ \lf| \f{d y_{i,k}^j}{ds} \ri|} \ri)^2  
\f{\lf| \f{d y_{i,k}^j}{ds} \ri|}{ \sum_{l=1}^N \sum_{h=1}^n \lf| \f{d y_{i,h}^{l}}{ds} \ri|}
\geqslant 
\lf(   \sum_{j=1}^N \sum_{k=1}^n \frac{1}{ \lf| \f{d y_{i,k}^j}{ds} \ri|} \ri)^2 .
\ee
By the convexity of the square function, inequality \eqref{atsui} holds true if
\be \label{tacchino}
\f{1}{4} \lf( \sum_{j=1}^N \sum_{k=1}^n  \ri)^2\geqslant 1.
\ee
Notice that $\sum_{j=1}^N \sum_{k=1}^n $ represents the number of solutions to the equation $\Phi=s$ for $a_{i-1}<s<a_i$ on the whole graph.
Since $\Phi$ is continuous and compactly supported, there are always at least two solutions and then \eqref{tacchino}
holds true. 
\end{dem}
\begin{teo}[P\'olya - Szeg\H o inequality]\mbox{ }
 \label{polya}
 
\n Assume that $\Phi \in \EE$ then $\Phi^\ast \in \EE $ and
\be \label{chiacchere}
\|{\Phi^{\ast}}'\|  \leqslant  \f{N}{2}    \|\Phi'\| \,.
\ee
\end{teo}
\begin{dem} 
Due to Proposition \ref{lp} $\Phi^*$ is symmetric, and  then continuous at the vertex. So that  it is sufficient to prove \eqref{chiacchere}. Let $\Phi\in\EE$ be positive and 
take $\Phi_n \in {\mathcal  PL}$ such that $\Phi_n \to \Phi$ in $H^1(\GG)$.
Take also a positive test function $\chi\in C_0^\infty (\GG )$.
Moreover, let $0=a_0< a_1< a_2 < \ldots $ be the values such that $\{ g= a_i \}$ has strictly positive measure. 
Notice that $g$ restricted to $\tilde D_i^\ast = \{ a_{i-1} < g < a_i\}$ is monotone and invertible by Lemma \ref{cont}. Monotonicity of $g$ also  implies that its derivative exists almost everywhere and is in $L^1_{loc}(\erre^+)$.     
Then the following inequalities hold:
\begin{align}
\lf|(\chi,{\Phi^{\ast}}')\ri| 
& =- \sum_{j=1}^N \int_0^{\infty} \chi_j(y) \, g'(y) \, dy \label{uno} \\
& = -\sum_{j=1}^N \sum_{i\geqslant 0}\int_{\tilde D_i } \chi_j(y) \, g'(y) \, dy \label{due} \\
& = \sum_{j=1}^N \int_0^{g(0)} \chi_j ( y^{\ast} (s) ) ds  \label{tre} \\
& \leqslant \liminf \sum_{j=1}^N \int_0^{g_n(0)} \chi_j ( y^{\ast}_n (s) ) ds \label{quattro}  \\
& = \liminf_n  - (\chi, {\Phi_n^{\ast}}' ) \label{cinque} \\
& = \liminf_n  \lf|(\chi, {\Phi_n^{\ast}}') \ri|  \label{cinqueprimo} \\
& \leqslant \liminf_n \| \chi \| \| {\Phi_n^{\ast}} ' \| \label{sei} \\
& \leqslant \frac{N}{2} \liminf_n \| \chi \| \| \Phi_n' \| \label{sette}\\
& = \frac{N}{2} \| \chi \| \| \Phi' \| .\label{otto}
\end{align}
The chain of inequalities stands for the following reasons. In
\eqref{uno} we explicitly wrote the
r.h.s.. In \eqref{due} we have restricted the integral to the regions where $g$ is not constant.
In \eqref{tre} we have changed variable of integration, the new one being $y^\ast (s)$
defined as before. This is possible due to the restriction made in
\eqref{due}. In \eqref{quattro} we have used Fatou's Lemma and  Lemma \ref{lemmolo}. 
In \eqref{cinque} we have changed back the integration variable and
in \eqref{cinqueprimo} we just changed a sign. In \eqref{sei} we have used Cauchy-Schwarz inequality.
In \eqref{sette} we have used lemma \ref{linear}. In \eqref{otto} we have used the convergence hypothesis.
Estimate \eqref{chiacchere} for a positive function $f$ follows from equation \eqref{otto} by Riesz Theorem.

\n
Now we extend the inequality to the general case. First notice that
Proposition \ref{cont} and  Proposition \ref{lemmolo} both extend to 
the non-positive case and to the complex valued case. A careful inspection of the argument used above, shows that
for positive $\chi \in C_0^\infty (\GG )$ it is  still valid until inequality \eqref{sei}. 
Then to conclude the proof we have to extend Proposition \ref{linear} to complex valued functions.

\n
For the real valued case, the extension is trivial. If $\Phi$ is piecewise linear then $|\Phi|\in {\mathcal PL} $ 
and has the same rearrangement of $\Phi$. Notice also that $|\Phi'|= | |\Phi|'|$ almost everywhere. Therefore since
 \eqref{tripoli} holds for $|\Phi|$, then it holds also for $\Phi$.

\n
Now for the complex valued case, to define the class of approximating functions  we set $\Phi = e^{i\Theta} F$ with  $\Theta, \,F \in {\mathcal PL}$, and where the product must be understood componentwise $e^{i\Theta} F = (e^{i\theta_1} f_1,...,e^{i\theta_N} f_N)^T$. This set is still dense in $\EE$.
Again notice that $\Phi^\ast = F^\ast$. We have also
\[
|\Phi'|^2 = |\Theta' F|^2 + | F'|^2.
\]
Therefore we can write
\[
\|{ \Phi^{\ast}}' \|^2 =
 \| {F^{\ast}}' \|^2 \leqslant
\| F' \|^2 \leqslant
\|\Theta' F\|^2 + \| F' \|^2 = 
\| \Phi' \|^2
\]
which proves equation \eqref{chiacchere} in the general case.
\end{dem}
\begin{remark}
The same  argument used to prove Theorem \ref{polya} can be used for the $W^{1,p}$ norm since $z \leadsto |z|^p$ is convex for $p\geqslant 1$.
\end{remark}
\begin{remark}
The constant $N^2/4$ is optimal. For instance take $\Phi$ such that 
\[
\phi_1 (y) = 
\begin{cases}
x & 0\leqslant x \leqslant 1 \\
2-x & 1\leqslant x \leqslant 2 \\
0 & x\geqslant 2 \\
\end{cases}
\qquad \qquad
\phi_i =0 \text{ for } i \neq 1
\]
Then $g$ can be easily computed by using \eqref{dragon}. One has:
\[
g (x) = 
\begin{cases}
1-\f{N}{2} x  & 0\leqslant y \leqslant \frac{2}{N} \\
0 & x \geqslant \frac{2}{N}  \\
\end{cases}
\]
From which
\[
\|\Phi'\|^2  = 2
\qquad
\|{\Phi^{\ast}}'\|^2  =  \f{N^2}{2}   \,.
\]
\end{remark}
\vskip10pt
\par\noindent
We end this Appendix with a comment on previous work on rearrangements on graphs contained in \cite{[Fri05]}. In \cite{[Fri05]}, the following P\'olya-Szeg\H{o} inequality has been proven for a function $\phi$ on a {\it bounded} graph with Kirchhoff conditions at vertices:
\begin{equation} \label{sua}
\| \phi^{\ast  \prime}   \| \leqslant \| \phi\|, 
\end{equation}
while here we proved for the function $\Phi$ an unbounded star graph (we prefer to use a different notation for making clearer the comparison) that 
\begin{equation} \label{nostra}
\| \Phi^{\ast \prime }\| \leqslant \frac{N}{2} \| \Phi\|.
\end{equation}

We would like to remark that both \eqref{sua} and \eqref{nostra} hold
true and are optimal: in fact, they refer to two different definitions of
rearrangements and to different boundary conditions at the free ends, which, as we will see, do matter. 
In the first place, notice that the rearranged functions defined in \cite{[Fri05]} are supported
on a segment (or, equivalently, on one edge of the graph), while the
rearranged functions defined in the present paper are symmetric with respect
to the exchange of edges and therefore they are supported on all
edges. 
Let us explain in details the origin of constants in the two settings.
Inequality \eqref{sua} was proved in \cite{[Fri05]} for a tree of finite total length $l$. 
The rearranged function $\phi^\ast$ is not  defined on the tree but it is defined on 
the segment $[0,l]$ and it is equimeasurable with $\phi$. On the other hand, $\Phi^\ast$ too
is equimeasurable with $\phi$, but as a function on the entire graph. Therefore,
the restriction of $\Phi^\ast$ to one edge, let us
call it $g$, is not equimeasurable with $\Phi$ and we have $N | \{ g >t\} | = | \{ \Phi >t\} |$.
As a consequence, if we compare $g$ to $\phi^\ast$, we see that $g$ goes to $0$ in a steeper way. 
This different normalization explains the $N$ in our estimate.

\par\noindent
Finally, a further dependence on boundary conditions has to be taken in account. In \cite{[Fri05]} the
author was interested in the eigenvalues of the Laplacian with
Kirchhoff conditions at vertices. In particular, for vertices of
degree $1$, i.e. free ends, this corresponds to Neumann boundary conditions.  The form
domain of this operator consists of functions which are $H^1$ on
edges, continuous at vertices of degree higher or equal than 2 and no
conditions at all at vertices of degree one.  Inequality \eqref{sua}
has been proved for this class of functions in lemma $3$ in
\cite{[Fri05]}.  In our case we have unbounded edges and we consider
$H^1$ functions on edges continuous at the vertex, and which of course
are vanishing at infinity.  In both proofs a key point is deriving a
lower bound for $n(t)$, defined as the number of solutions of
$\phi(x)=t$, uniform in each class of functions.  We have the lower
bound $n(t)\geqslant 2$ while in \cite{[Fri05]}, see equation (2.5), the
lower bound is $n(t)\geqslant 1$.  This difference explains the
factor $2$ appearing in the denominator of \eqref{nostra} and missing in \eqref{sua}.  We
think that both estimates are optimal for the two different
geometrical settings.  In the case studied in \cite{[Fri05]}, one could consider a positive
function, starting from an endpoint of the graph, localized on one
single edge and vanishing in a monotone way. For such a function we
have $n(t)=1$ and therefore $n(t)\geqslant 1$ is optimal. For our
admissible functions such a behavior is impossible since we have
functions going to 0 at infinity and globally continuous. So we cannot have better than $n(t)\geqslant 2$.  We think
that it would be natural to compare a star graph with infinite length
with a star graph of finite length but Dirichlet boundary conditions
in the end points. For such a graph we expect P\'olya-Szeg\H{o}
inequality to take the form $\| \phi^{\ast \prime} \| \leqslant
\frac12 \| \phi\|$.
To conclude, several definitions for a rearrangement on a graph can be given,
and moreover the optimal constant in the P\'olya-Szeg\H{o} inequality depends in a
sensible way from the chosen definition and from the boundary conditions at the free vertices.

Our choice was natural in the geometrical setting of this model.
The presence of a central point of the star graph, i.e. the vertex, makes natural to define the rearranged function to be 
symmetric w.r.t. to this point as, in facts, one does in the $\erre^n$ case. 
Moreover, in this way the rearranged function is still defined on the star graph and this gave us
intuition on the minimizers.

\n

\section{Useful identities}
\label{appB}
In this section we recall some useful identities that will be used several times in the paper.
We label the soliton profile on the real line as
\be \label{soliton}
\phi_s (x) = [ (\mu + 1) \ome]^{\frac{1}{2\mu}} \sech^{\frac{1}{\mu}} (\mu \sqrt{\ome} x).
\ee
It satisfies the equation
\be \label{soliteq}
-\phi_s '' -|\phi_s |^{2\mu} \phi_s = -\ome \phi_s. 
\ee
Moreover, multiplying by $\ove \phi_s $ and integrating one checks that
\[
\| \phi_s ' \|^2_{L^2(\erre)} - \| \phi_s  \|^{2\mu +2}_{L^{2\mu+2} (\erre)} +\ome \|  \phi_s  \|^2_{L^2(\erre)} =0.
\]
Starting from  definition \eqref{soliton} and changing variables in the integrals, one obtains the following formulas:
\begin{align}
&\int_0^\infty |\phi_s (x+\xi )|^2 dx \quad= \frac{(\mu+1)^{\frac{1}{\mu} }}{\mu} \ome^{\frac{1}{\mu} -\frac{1}{2} } 
\int^1_{\tanh (\xi \mu\sqrt{\ome})} (1-t^2)^{\frac{1}{\mu} -1} dt \label{formula1} \\
&\int_0^\infty |\phi_s (x+\xi )|^{2\mu+2} dx = \frac{(\mu+1)^{1+\frac{1}{\mu} }}{\mu} \ome^{\frac{1}{\mu} +\frac{1}{2} } 
\int^1_{\tanh (\xi \mu\sqrt{\ome})} (1-t^2)^{\frac{1}{\mu} } dt \label{formula2} \ .
\end{align}

\end{document}